\documentclass[11pt]{article}
\usepackage{setspace}
\usepackage[OT1]{fontenc}
\usepackage{amsthm,amsmath,amssymb}
\usepackage{graphicx}
\usepackage{multirow}
\usepackage{multicol}
\usepackage{fullpage}
\usepackage{nonfloat}
\usepackage{url}
\usepackage{smile}
\usepackage{pifont}
\usepackage{pdflscape}
\usepackage{placeins}
\usepackage{natbib}
\usepackage[colorlinks,
            linkcolor=red,
            anchorcolor=blue,
            citecolor=blue
            ]{hyperref}

\numberwithin{equation}{section}
\numberwithin{theorem}{section}

\begin{document}

\title{\huge  {Statistical Depth for Big Functional Data with Application to Neuroimaging} 
}

\author{
Alicia Nieto-Reyes\thanks{Department de Mathematics, Statistics and Computer Science, University of Cantabria, Santander, Spain, E-39005 (E-mail: \emph{alicia.nieto@unican.es}).} \thanks{The authors gratefully acknowledge the help of Ci-Ren Jiang with the implementation of the FPCA deconvolution. A. N-R. also thanks the University of Cambridge for kind hospitality during the time this work was carried out. The research of A. N-R. was supported in part by the Spanish Ministry of Science, Innovation and Universities, grant MTM2017-86061-C2-2-P, and the MECD (Spain), grant CAS14/00375. The research of J.A. was supported by Engineering and Physical Sciences Research Council grant EP/K021672/2.} \and   John A. D. Aston \thanks{Statistical Laboratory, University of Cambridge, Wilberforce Road, Cambridge CB3 0WB, UK}
}



\maketitle

\begin{abstract}
Functional depth is the functional data analysis technique that orders a functional data set. Unlike the case of data on the real line, defining this order is non-trivial, and particularly, with functional data, there are a number of properties that any depth should satisfy. We propose a new depth which both satisfies the properties required of a functional depth but also one which can be used in the case where there are a very large number of functional observations or in the case where the observations are functions of several continuous variables (such as images, for example). We give theoretical justification for our choice, and evaluate our proposed depth through simulation. We finally apply the proposed depth to the problem of yielding a completely non-parametric deconvolution of Positron Emission Tomography (PET) data for a very large number of curves across the image, as well as to the problem of finding a representative subject from a set of PET scans.
\end{abstract}


\section{Introduction}\label{sectionIn}

Statistical depth is the technique that orders the elements of a space with respect to a dataset or distribution. This order is trivial on the real line but much less so in higher dimensions. For multivariate spaces, the notion of depth was formalized in \citet{Zuo00}, based on the properties defined in \citet{Liu90}, and prominent examples include the Tukey depth \citep{Tukey} and the simplicial depth \citep{Liu90}. For the last twenty years there have been several extensions of these definitions to cover high dimensional spaces and functional spaces (cf.~Section \ref{sectionOdFD}) and in \cite{NRBattey2015} the notion was formalized for functional spaces. There, six properties were established, which are needed for a functional depth to be useful. The first depth function to satisfy these properties was recently  prosed in \cite{Nieto21a}, and applied to real data in \cite{Nieto21b}. However, it is computationally expensive.
Thus, the aim of this paper is to present an example of statistical functional depth that is notable in various ways: It satisfies the six properties notion of functional depth and can be used for \textit{big} functional data in an efficient manner, where here, \textit{big} is used to mean any data where computational considerations are amongst the most pressing. Additionally, to our knowledge, this paper presents the first
 application of functional depth to raw functional data that are simultaneously continuous in several dimensions and have each datum observed in a different discretization subset of the domain, the importance of which will be demonstrated in its use in neuroimaging.

Overall, a statistical depth function takes its deepest value at a measure of location; commonly a generalized median although the literature includes a case of a generalized mode \citep{Cuevas07}. We restrict ourselves here to those taking deepest value at a generalized median for its suitability to our neuroimaging application.
These existing 
examples of functional depth  \citep[for instance]{ChakrabortyAoS2014, randomTukey, Romo09, Romo11} either do not satisfy the six properties of the notion of functional depth in their entirety or are infeasible when applied to \textit{big} data.  A property that is rarely satisfied by existing examples of functional depth regards the receptivity to the convex hull width across the domain of the depth function. This property essentially means that the notion of depth should be somewhat invariant to measurement errors in domains where there is little to no variation in the underlying functions. In fact, only one of these examples of depth is
known to satisfy this property, the metric depth \citep{Nieto21a}, but that depth is computationally infeasible for large amounts of data.
Focusing on the left plot of Figure \ref{Oc}, the idea behind this property is that the order of the curves in the region of the domain around the time corresponding to 4000 in the x-axis, where the data display little variability and significantly overlap with one another, should not carry undue influence in the overall order of the data provided by the depth. The left plot of Figure \ref{Oc} represents $100$ curves selected at random out of a dataset of $N=254,807$ curves of Positron Emission Tomography (PET) data. PET, the motivating application of the work in this paper, is a neuroimaging modality which can be used to investigate neurochemistry \textit{in-vivo}. A scan specific input injection of radiotracer is delivered to the brain via the blood plasma, and the resulting positron decay recorded by the scanner. In the right plot of Figure \ref{Oc}, the curves in the left plot are plotted in grey and in color the deepest curves amongst the $N$ according to different well-known notions of functional depth, including the proposed notion in this paper, are shown.


\begin{figure}[htb]
\begin{center}
\includegraphics[height=7cm,width=.4\linewidth]{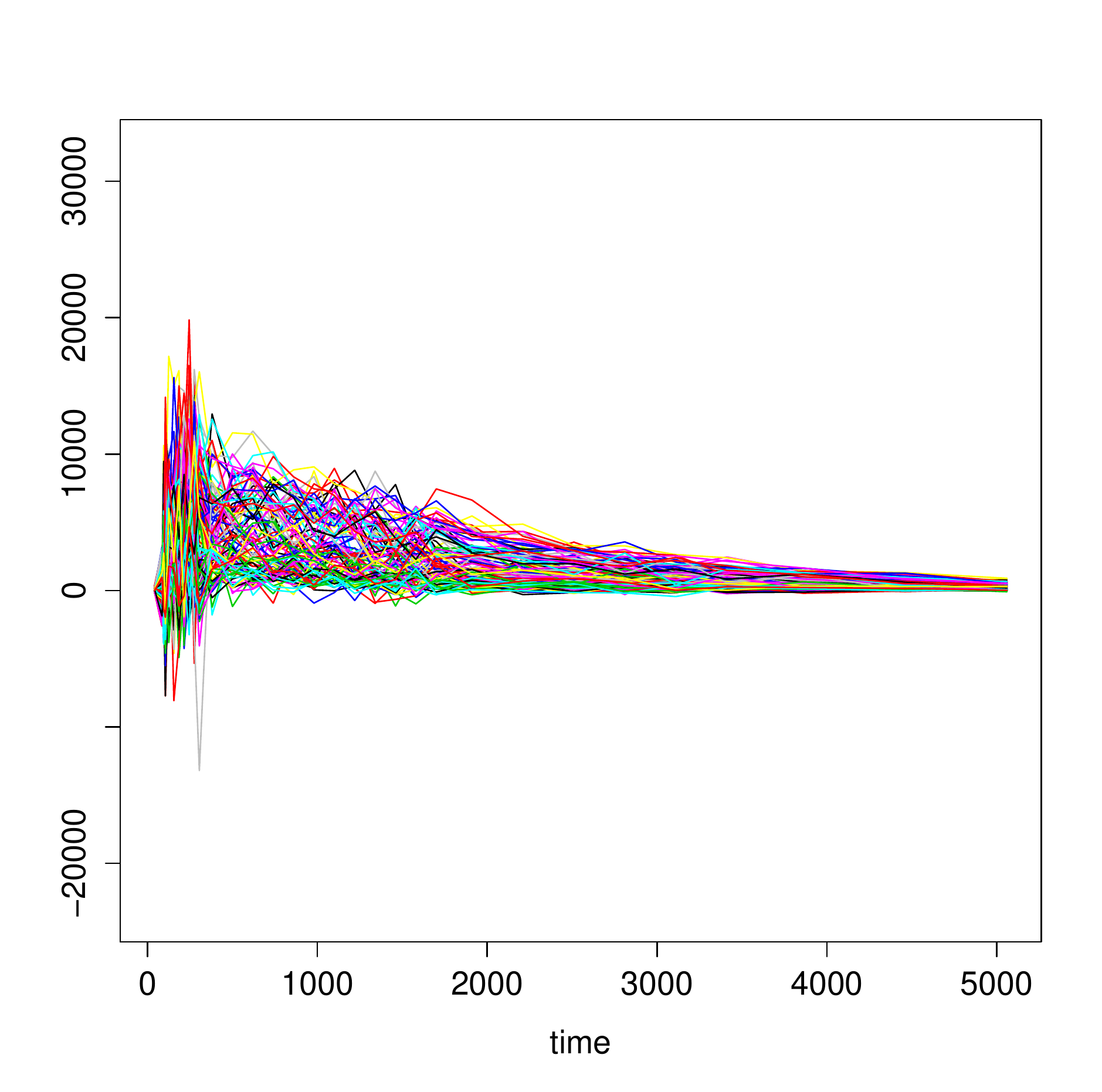}
\includegraphics[height=7cm,width=.4\linewidth]{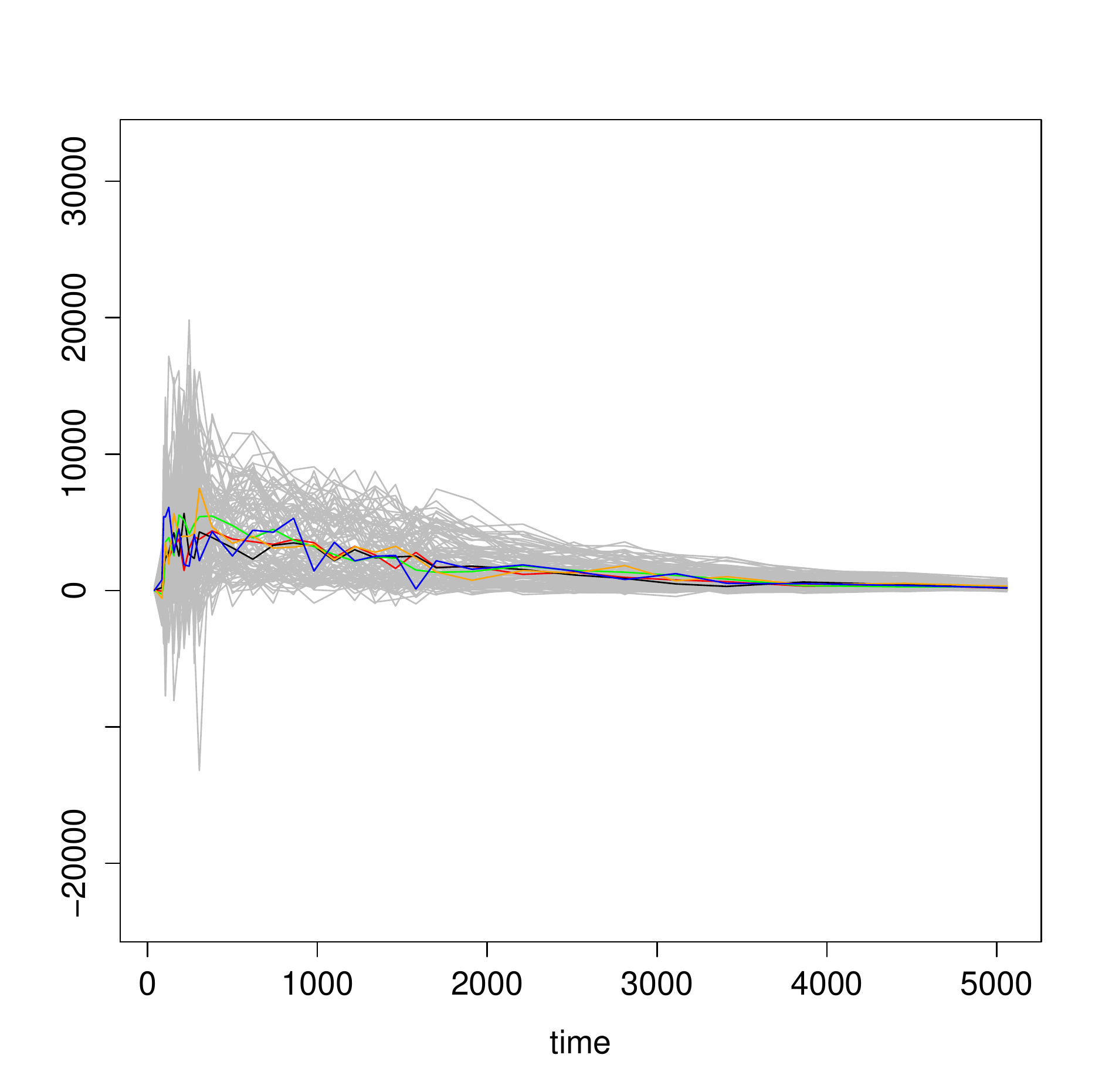}
\end{center}
\caption{$100$ curves selected at random out of a dataset of $N=254,807$ PET curves are represented in colour (left) and in grey (right). The right plot also represents the deepest curve among the $N$ according to different examples of functional depth. $D_I$ (red), $D_M$ (black), $D_B$ (green), $D_{T}$ (orange) and $D$ (blue) (cf. Sections \ref{sectionFD} and \ref{sectionOdFD}).}
\label{Oc}
\end{figure}

The proposed example of depth is defined for any semi-metric space as a random version of the metric depth, a functional depth based on the distances of the curves to a center of functional symmetry (cf.~Section \ref{sectionFD}). Data from PET typically consists of large four-dimensional volumes, where images of the brain are collected over time in an irregularly spaced fashion. This results in functional data which can be viewed in a number of ways; (i) as 1-D time curves collected over space; (ii) as 3-D functional image volumes collected over time;  or (iii) as full 4-D data collected over a number of subjects. Therefore, the depth we introduce need to be capable of dealing with both big functional data, i.e.~as in (i), and functional data that is simultaneously continuous in different dimensions, as in (ii) and (iii). For the first,  we propose an algorithmic definition to effectively compute the proposed random depth function. For the second, taking into account that data in (iii) has the 
generality that the subset of the domain in which each datum is partially observed differs for each datum,  in the present paper we construct an appropriate semi-metric to apply in this domain, the metric depth or the proposed definition when dealing with big data. In the literature there has yet to be a functional depth proposed that applies to the setting of the data being simultaneously continuous in different dimensions with each datum being observed in a different discretization subset of the domain without having to preprocess the data previously; although the random Tukey depth  \citep{randomTukey} could be applied if it were known that there existed a coherent distribution in which to project the data, as occurs in the multivariate case.
Using the proposed depth, this paper has two particular applications of functional depth in PET analysis. Firstly, it is used to investigate the deconvolution of dynamic neuroimaging data \citep{FPCA} and, second, to choose a representative subject of a dynamic neuroimaging dataset. Both applications are fundamental in the preprocessing of PET data (cf.~Section \ref{ARD}) and employing the proposed statistical depth function results in a first truly non-parametric preprocessing of PET data.

Section \ref{Ab} contains the required notation and background. 
To achieve these ends, in Section \ref{sectionFD} this paper proposes a new functional depth that emerges from the notion of metric depth. Although this depth is  proposed here in the context of neuroimaging data, it can be used in any framework. From a theoretical point of view, it is proved  here to be the first computationally feasible example to satisfy the definition of functional depth given in \cite{NRBattey2015}.  The proposed functional depth is shown to perform consistently well in comparison to any of the existing examples of functional depth, see Section \ref{sectionOdFD}. In Section  \ref{Simu} we run a series of simulations to investigate the performance of the proposed notion of functional depth for big data and its performance under noisy data. From a practical point of view, the proposed notion of depth is compared with the existing alternatives in the context of neuroimaging in Section \ref{ARD}, showing good performance of the proposed notion. Finally, a short conclusion is given in Section \ref{sectionDisc}. The computer code is available upon request.



\section{Notation and background on functional depth}\label{Ab}


$\mathfrak{F}$ denotes a functional metric space with associated metric $d(\cdot,\cdot)$ and $I\subseteq \RR^p$  a compact subset that is the domain of definition of the elements in $\mathfrak{F}.$
 $P$ denotes a distribution in $\mathcal{P},$ the space of all probability measures on $\mathfrak{F},$  and $\mathcal{S}:=\mathcal{S}_P$ the support of $P.$  
 $(\mathfrak{F},\mathcal{A},P)$ represents a probability space  and 
$\mathfrak{C}(\mathfrak{F},P)$ the convex hull of $\mathfrak{F}$ with respect to $P$ as defined in  \citet{NRBattey2015}.

As commented in the introduction, a main  objective of this paper is to find a functional depth that satisfies the six defining properties of statistical functional depth.
 According to \citet{NRBattey2015, rCorrectS}, the mapping
$D(\cdot,\cdot):\mathfrak{F}\times\mathcal{P}\longrightarrow \mathbb{R}$ is a \emph{statistical functional depth} if it satisfies
properties P-1. to P-6., below.

P-1.  \emph{Distance invariance}. $D(f(x),P_{f(X)})=D(x,P_{X})$ for any $x\in\mathfrak{F}$ and $f:\mathfrak{F}\rightarrow\mathfrak{F}$ such that for any $y\in\mathfrak{F}$, $d(f(x),f(y))=a_f \cdot d(x,y)$, with $a_f>0$.

P-2. \emph{Maximality at centre}.
For any $P\in\mathcal{P}$ possessing a unique centre of symmetry  $\theta\in\mathfrak{F}$ with respect to some notion of functional symmetry, $D(\theta,P)=\sup_{x\in\mathfrak{F}}D(x,P)$.
%

P-3. \emph{Strictly decreasing with respect to the deepest point}. 
For any $P\in\mathcal{P}$ such that $D(z,P)=\max_{x\in\mathfrak{F}}D(x,P)$ exists with $D(z,P)=D(z',P)$ implying $d(z,z')=0,$ $D(x,P)< D(y,P)<D(z,P)$ holds for any $x,y\in\mathfrak{F}$  such that 
$\min\{d(y,z),d(y,x)\}>0 \mbox{ and }  \max\{d(y,z),d(y,x)\}<d(x,z).$

P-4. \emph{Upper semi-continuity in $x$}. $D(x,P)$ is upper semi-continuous as a function of $x,$ i.e., for all $x\in\mathfrak{F}$ and for all $\epsilon>0$, there exists a $\eta>0$ such that 
%
$\sup_{y \in\mathfrak{F}_{x} \; : \; d(x,y)<\eta}D(y,P)\leq D(x,P)+\epsilon,$
%
where $\displaystyle{\mathfrak{F}_{x}:=\Bigl\{y\in \mathfrak{F}: d(y,x)<d(y,\theta) \text{ or }\max\{d(y,\theta),d(y,x)\}<d(x,\theta)\Bigr\}}$ for $\theta
\in
\arg\sup_{x\in\mathfrak{F}}D(x,P)$.

P-5. \emph{Receptivity to convex hull width across the domain.} $D(x,P_X)<D(f(x),P_{f(X)})$ for any $x\in \mathfrak{C}(\mathfrak{F},P)\backslash 0$ with $D(x,P)<\sup_{y\in\mathfrak{F}}D(y,P)$ and $f:\mathfrak{F}\rightarrow\mathfrak{F}$ such that $f(y(v))=\alpha(v) y(v)$ with $\alpha(v)\in (0,1)$ for all $v\in L_\eta$ and $\alpha(v)=1$ for all $v\in L^c_\eta.$ 
$
L_{\eta}:=\arg\sup_{H\subseteq I} \Bigl\{|H|: \sup_{x,y\in\mathfrak{C}(\mathfrak{F},P)} d(x(H),y(H))\leq\eta \Bigr\}
$
for any $\eta\in[\inf_{v\in 
I
} d(L(v),U(v)),d(L,U))\cap(0,\infty)$ such that $\lambda(L_{\eta})>0$ and $\lambda(L^c_{\eta})>0,$ with $|H|$ denoting the length of $H.$

P-6. \emph{Continuity in $P$}. For all $x\in\mathfrak{F}$, for all $P\in\mathcal{P}$ and for every $\epsilon>0$, there exists a $\eta(\epsilon)>0$ such that $|D(x,Q)-D(x,P)|<\epsilon$ $P$-almost surely for all $Q\in\mathcal{P}$ with $d_{\mathcal{P}}(Q,P)<\eta$ $P$-almost surely, where $d_{\mathcal{P}}$ metricises the topology of weak convergence.

Instead of P-2, in \cite{NRBattey2015} it is used the following property.

P-2G. \emph{Maximality at Gaussian process mean}.
For $P$ a zero-mean, stationary, almost surely continuous Gaussian process on $I$, $D(\theta,P)=\sup_{x\in\mathfrak{F}}D(x,P)\neq\inf_{x\in\mathfrak{F}}D(x,P)$, where $\theta$ is the zero mean function.

P-2. regards a symmetric distribution and a center of symmetry of it, so we recall the one used in  \cite{Nieto21a}.
 A probability distribution $P$ on $\mathfrak{F}$ with support $\mathcal{S}$ is  \emph{$d$-symmetric}  about $z:=z(P)\in \mathcal{S}$ if  
 \begin{eqnarray}\label{cfhs}
 P(\mathfrak{H}^{z}_x)\geq\frac{1}{2} \mbox{ for all } x\in\mathcal{S}, 
 \end{eqnarray}
 where 
$ 
\mathfrak{H}^{z}_x:=\mathfrak{H}^{z}_x(P)=\bigl\{y\in \mathcal{S}: d(y,x)\geq\max\{d(x,\theta),d(y,\theta)\}\bigr\}.
$ 
 Additionally, denoting
 \begin{eqnarray} \label{cs}
\Theta:=\Theta(P)=\{z\in \mathcal{S}: P(\mathfrak{H}^{z}_x)\geq \frac{1}{2}\mbox{ for all }x\in \mathcal{S}\},
 \end{eqnarray}
any $ \theta\in\Theta$  is a  \emph{centre of $d$-symmetry} of $\mathfrak{F}$  with respect to $P.$ 
Substituting the inequality in \eqref{cfhs} by $P(\mathfrak{H}^{z}_x)\geq \frac{1}{2}-\delta$ with $\delta\in [0,\frac{1}{2}]$ results in the notion of \emph{$(d,\delta)$-symmetry}; and the analogous $\theta\in\Theta$ in  \eqref{cs} is respectively referred as a  center of $(d,\delta)$-symmetry.

Our  depth proposal is based on the following. 
Let $P$ be a probability distribution on $\mathfrak{F}$ with  a non-empty set of centers of $(d,\delta)$-symmetry $\Theta.$  The \emph{metric depth} at $x\in\mathfrak{F}$ with respect to $P$ is 
$ 
D_{\mbox{m}}(x,P)=\{I_{m}(x,P) + 1\}^{-1}, \mbox{ where }
I_m(x,P):=    d(x,\Theta)/d(\vartheta,\vartheta'),
$ 
$\vartheta:=    \vartheta(P),\vartheta':=   \vartheta'(P)\in\mathfrak{F}$ independent of $x$ with $d(\vartheta,\vartheta')>0.$
Note that  
\begin{eqnarray}\label{dC} 
d(x,\Theta):= \inf_{\theta\in\Theta}d(x,\theta).
\end{eqnarray}
The metric depth  satisfies  the  axiomatic definition for  the below conditions, on the selection of $\vartheta,\vartheta'\in\mathfrak{F}$
and   on the metric. 
\begin{assumption}
\label{condV} 
For all $\varepsilon>0,$ there exists an $\eta>0$ satisfying
$ 
|d\{\vartheta(P),\vartheta'(P)\}-d\{\vartheta(Q),\vartheta'(Q)\}|<\varepsilon\mbox{ for all }Q\in\mathcal{P}\mbox{ with }d_{\mathcal{P}}(Q,P)<\eta.
$ 
\end{assumption}
\begin{assumption}
\label{condD} Let $\mathcal{B}\subset I$ of positive Lebesgue measure and $g: \mathfrak{F}\rightarrow \mathfrak{F}$ such that for any $x\in\mathfrak{F}$ and $v\in I,$ $g(x(v))=\alpha(v)x(v)$ where $\alpha(v)\in(0,1]$ with $\alpha(v)<1$ for all $v\in\mathcal{B}$. For any $\mathcal{B}$ and $g$ as above,  $d(x,y)>d\{g(x),g(y)\}$  for all $x,y\in\mathfrak{F}$ with $x(v)\neq y(v)$ for all $v\in\mathcal{B}.$  
\end{assumption}

\section{Methodology}\label{sectionFD}

In dealing with big functional data, our proposal, the \emph{random depth},  computes the distance with respect to a random center of 
symmetry.  The reason for using a random center is that it can be computed extremely fast while the resulting random depth satisfies the six defining properties of functional depth, Theorem \ref{Td}, and converges to the metric depth, Theorem \ref{DMD}. 
All the definitions are written in terms of metric spaces for ease of reading, however, analogous definitions can be provided in terms of pseudo-metric spaces. After establishing the theoretical properties of the proposed random depth (Section \ref{Tp}), we devote Subsection \ref{FMD} to elaborate on the random depth, and the metric depth, in the context of multiple functional dimensions.

\subsection{Definition of random depth}



\begin{definition}[Random center of symmetry]\label{rc}
Let $(\mathfrak{F},d)$ be a functional metric space, $P$ a probability distribution on $\mathfrak{F}$ with support $S$ and $\{P_n\}_n$ and $\{P_m\}_m$ two independent sequences of empirical distributions of $P.$ 
   Let us denote $$\Theta_{n,m}:=\Theta_{n,m}(P)=\{z\in S: P_n(\mathfrak{H}^{z}_{x})\geq \frac{1}{2} \mbox{ for all }x\in S_m\}$$
   where   $S_m$ is the support of $P_m.$
Then, any $\theta_{n,m}\in\Theta_{n,m},$ if it exists,  is a \emph{random center of symmetry} of $\mathfrak{F}$ with respect to $P$ based on $P_n$ and $P_m.$
 \end{definition}


The existence of at least a random center of symmetry corresponds to $P$ being randomly symmetric.

\begin{definition}[Random symmetry]\label{rs}
Let $(\mathfrak{F},d)$ be a functional metric space, $P$ a probability distribution  on $\mathfrak{F}$ with support $S$ and $\{P_n\}_n$ and $\{P_m\}_m$ two independent sequences  of empirical distributions of $P.$
Then,  $P$ is \emph{randomly functional  symmetric}, or just \emph{randomly symmetric,} about $\theta_{n,m}
\in S$  
based on $P_n$ and $P_m,$ 
if
$\displaystyle P_n(\mathfrak{H}^{\theta_{n,m}}_{x})\geq 1/2,  \mbox{    for all    } x\in S_m,$ with $S_m$ the support of $P_m.$
\end{definition}

The use of the name random functional symmetry arises because our random proposal is derived from the notion of  $(d,\delta)$-symmetry  \citep{Nieto21a} that is intended for functional spaces.
Substituting
the inequality in Definition \ref{rs} by $P_n(\mathfrak{H}^{\theta_{n,m}}_{x})\geq 1/2-\delta$ with $\delta\in [0,\frac{1}{2}]$ results in the notion of \emph{$\delta$-random symmetry}. We define  the analogous $\theta_{n,m}$ of Definition \ref{rc}  as a  \emph{center of $\delta$-random symetry}. Note that a distribution $P$ is always $\delta$-randomly symmetric for some $\delta\in [0,\frac{1}{2}].$ 
Definition  \ref{rd}, below,  also applies in the case of  $\theta_{n,m}$ being the center of $\delta$-random symmetry of $P.$ Consequently, a random depth can be computed  without any imposition on $P.$ As the distributions studied in this paper are randomly symmetric (cf.~Section \ref{ARD}), it suffices for us to work with the notion of random symmetry.

\begin{definition}[Random depth]\label{rd}
Let $(\mathfrak{F},d)$ be a functional metric space and $P$  a 
probability distribution on $\mathfrak{F}$ with  a random center of symmetry $\theta_{n,m}.$   The \emph{random depth} at $x\in\mathfrak{F}$ with respect to $P$ based on  $\theta_{n,m}$  is
$$D(x,P):=\left[1+\displaystyle\frac{d(x,\theta_{n,m})}{d(\vartheta_{n,m},\vartheta_{n,m}')}\right]^{-1},$$
where
$\vartheta_{n,m}:=    \vartheta_{n,m}(P),\vartheta_{n,m}':=   \vartheta_{n,m}'(P)\in\mathfrak{F}$ are independent of $x$ with $d(\vartheta_{n,m},\vartheta_{n,m}')>0.$
 \end{definition}
 $\Theta_{n,m}$  can have more than one element, thus, the random depth depends on the element selected by the practitioner. 
 The definition also uses  a distance as well as two functions of the distribution, $\vartheta$ and $\vartheta'.$ Their selection is key in the depth satisfying the six properties, as we see in the next subsection. However, many possibilities apply. One of them is established in the following assumption, where $\vartheta$ is as in the metric depth.
 \begin{assumption}\label{A3}
$\vartheta_{n,m}'=\theta_{n,m}$ and $\vartheta_{n,m}$ is a function of $\Theta_{n,m}$ such that $d(\vartheta,\vartheta_{n,m})\rightarrow_{P} 0.$ 
 \end{assumption}
The random depth takes value one at the random center of symmetry with respect to which it is computed and decreases for values farther away from the center in terms of the distance used. The function can be easily transform to, for instance, take value zero at the random center of symmetry and increase for values farther away from the center in terms of the distance used.
 In the next subsection,  we will see that not only is the random depth a computationally effective approximation of the metric depth, but also a statistical functional depth on its own right.

A more straightforward definition for random depth is, however,
$$D_r(x,P):=[1+d(x,\theta_{n,m})]^{-1},$$
though this is in general not a valid definition because, as we see in next subsection,  it does not satisfy property P-1.~in the definition of statistical functional depth. However, when the application under study involves only one probability distribution its use is equivalent; as it occurs to the metric depth \citep{Nieto21b}.
\subsection{Theoretical properties}\label{Tp}

Given $P_m$ an empirical distribution, of a distribution $P,$ with support $S_m,$ let us use the notation $\Theta_{,m}:=\{z\in S:P(\mathfrak{H}^z_{x})\geq 1/2$ for all $x\in S_m\}.$ 
Using \eqref{cs}, it is worth noticing that $\Theta_{,m}\neq\emptyset$ if  $P$ is a   $d$-symmetric distribution. This is due to $\Theta\subseteq\Theta_{,m}.$   As  $\Theta_{,m}\supseteq \Theta_{,m+1}\supseteq \cdots \supseteq \Theta,$ we have that $\theta\in\Theta_{,m}$ for any $\theta$ center of $d$-symmetry of $\mathfrak{F}$ with respect to $P$ and that $\lim_{m\rightarrow\infty}\Theta_{,m}=\Theta$ $P$-almost surely. $\Theta_{n,m}$ can be empty though. However,  the larger  that $P(\mathfrak{H}^z_{x})$ is than .5, for a given $z$ and $x,$ and the larger $n$ is, the more probable that   $P_n(\mathfrak{H}^z_{x})$ is larger than .5; and consequently, the more probable that $\Theta_{n,m}$ is non-empty. 

The properties established in this section are important from a theoretical point of view and in practice for our later application in neuroimaging. The proofs of the results, however, follow smoothly from existing results in the literature.
Next Theorems \ref{td} and \ref{prima} entails an abuse of notation in that we also write ${d_{H}\{\Theta_{n,m},\Theta\}\rightarrow_{p} 0}$ when $d_{H}$ is constructed on empty sets.

\begin{theorem}\label{td} Let $(\mathfrak{F},d)$ be a functional metric space and $P$ a probability distribution on $\mathfrak{F}$ with compact support $S.$ 
 Then, $$d_{H}\{\Theta_{n,m},\Theta\}\rightarrow_{P} 0,$$ where, for arbitrary sets $A,$ $B\subseteq\mathfrak{F},$  $d_{H}(A,B)=\max\Bigl\{\sup_{a\in A} d(a,B), \sup_{b\in B} d(b,A)\Bigr\},$ with $d(a,B)$ as in \eqref{dC}.
  \end{theorem}
 The assumption of compactness of $\mathcal{S}$ is made for technical convenience. This assumption can be relaxed by controlling the probability of $X$ outside of a compact subset of $\mathcal{S}$.  Although different, this proof follows along the lines of \citet[Theorem 2]{Nieto21a}.
 \begin{proof}
 By the triangular inequality, we have that $$d_H(\Theta,\Theta_{n,m})\leq d_H(\Theta,\Theta_{,m})+  d_H(\Theta_{,m},\Theta_{n,m}).$$ 
To prove that
$d_{H}\bigl(\Theta,\Theta_{,m}\bigr)\rightarrow_P 0,$ we  use  the compactness of $\mathcal{S}$, which ensures $d_{H}\bigl(S_{m},S\bigr)\rightarrow_P 0$. Since $\mathcal{S}_{m} \subseteq \mathcal{S}$, we have that $\Theta \subseteq \Theta_{,m}$, with equality if and only if  
$$\{x\in \mathcal{S} \cap \mathcal{S}_{m}^{c}: P(\mathfrak{H}_{x}^{z})< \frac{1}{2}  \mbox{ for all } z\in  \Theta_{,m}\}=\emptyset.$$ We thus have $d_{H}(\Theta, \Theta_{,m}) \leq d_{H}(\mathcal{S},\mathcal{S}_{m}) \rightarrow_{P} 0$ as  $m\rightarrow\infty.$

For the control over $d_H(\Theta_{,m},\Theta_{n,m}),$
as the ball $\sigma$-algebra coincides with the Borel $\sigma$-algebra on separable semi-metric spaces \cite[chapter 1.7]{rvv}, weak convergence of $P_{n}$ to $P$ gurantees that $P_{n}(\mathfrak{H}_{x}^{z})\rightarrow P(\mathfrak{H}_{x}^{z})$ for any $x\in\mathcal{S}_{m}$ as $n \rightarrow \infty$, 
resulting in $d_H(\Theta_{,m},\Theta_{n,m})\rightarrow_{P} 0$ as $n\rightarrow\infty.$ Consequently,  $d_H(\Theta,\Theta_{n,m})\rightarrow_{P} 0.$
 \end{proof}

 A consequence of the proof is the following result.
 \begin{corollary}
 In the setting of Theorem \ref{td}, $d_{H}\{\Theta_{n,m},\Theta_{,m}\}\rightarrow_{P} 0$ 
 and $d_{H}\{\Theta,\Theta_{,m}\}\rightarrow_{P} 0.$
 \end{corollary}
 
 Using a random procedure means that each time a center of symmetry is computed a different result is generally obtained. However, as we see in below Corollary \ref{prima}, the distance amongst two sets of centers of symmetry, based on independent draws from the distribution at hand, converges to zero. The use of random procedures is not novel in the literature of data depth \citep{randomTukey}, nor in general in the  literature of high-dimensional, or functional, data \citep{Samworth}.
 \begin{corollary}\label{prima} In the setting of Theorem \ref{td}, let $\{P_n\}_n,$  $\{P_m\}_m,$ $\{P'_n\}_n $ and $\{P'_m\}_m$ be four independent sequences  of empirical distributions of $P,$ with associated subsets $\Theta_{n,m}$ and $\Theta'_{n,m}.$
 Then, $d_{H}\{\Theta_{n,m},\Theta'_{n,m}\}\rightarrow_{P} 0.$
  \end{corollary}

Theorem \ref{DMD} establishes the convergence of the random depth to the metric depth. 
\begin{theorem}\label{DMD}
Let $(\mathfrak{F},d)$ be a functional metric space, $P$ a $d$-symmetric probability distribution on $\mathfrak{F}$ and $x\in\mathfrak{F}.$   Let us  assume there exists a sequence of non-empty compact sets $\{\Theta_{n,m}\}_{n,m}.$ Then, there exists a sequence $\{\theta_{n,m}\}_{n,m},$ with $\theta_{n,m}\in\Theta_{n,m}$ for all $(n,m)\in\mathbb{N}\times\mathbb{N}$ such that  $$d(\theta_{n,m},\Theta)\rightarrow_P 0\mbox{ and  }
|D(x,P)-D_{\mbox{m}}(x,P)|\rightarrow_P 0,$$ where D is computed with respect to $\theta_{n,m}$ and $\vartheta_{n,m}$ and $\vartheta_{n,m}'$ fulfill Assumption \ref{A3}.
\end{theorem}

 \begin{proof}
 For each $(n,m)\in\mathbb{N}\times\mathbb{N},$ let us fix $\theta_{n,m}\in\arg\inf_{z\in\Theta_{n,m}}d(\Theta,z)$ and $\vartheta'=\theta\in\arg\inf_{z\in\Theta}d(x,z).$
%
We have that $0\leq d(\Theta,\theta_{n,m})\leq d_H(\Theta,\Theta_{n,m})$ and $0\leq d(\theta,\theta_{n,m})\leq d_H(\Theta,\Theta_{n,m}),$ by the definitions of $\theta,$ $\theta_{n,m}$ and $d_H(\cdot,\cdot).$ Then, as $d_H(\Theta,\Theta_{n,m})\rightarrow_P 0$ by Theorem \ref{td}, $d(\Theta,\theta_{n,m})\rightarrow_P 0$ and  \begin{eqnarray}\label{dt}d(\theta,\theta_{n,m})\rightarrow_P 0.\end{eqnarray}

By the triangular inequality, $d(\theta_{n,m},\vartheta_{n,m})\leq d(\theta_{n,m},\theta)+d(\theta,\vartheta)+d(\vartheta, \vartheta_{n,m})$ and $d(\theta,\vartheta)\leq d(\theta,\theta_{n,m})+d(\theta_{n,m},\vartheta_{n,m})+d(\vartheta_{n,m}, \vartheta).$ Thus, $|d(\theta_{n,m},\vartheta_{n,m})-
d(\theta,\vartheta)|\leq d(\theta,\theta_{n,m})+d(\vartheta_{n,m}, \vartheta).$ 
Using then  (\ref{dt}) and Assumption \ref{A3}, 
we get  
\begin{eqnarray}\label{v}|d(\theta_{n,m},\vartheta_{n,m})-d(\theta,\vartheta)| \rightarrow_P 0.\end{eqnarray}  
Let us denote $R:=\left|
\displaystyle\frac{d(x,\theta)}{d(\vartheta,\theta)}
-\displaystyle\frac{d(x,\theta_{n,m})}{d(\vartheta_{n,m},\theta_{n,m})}
\right|\cdot d(\theta_{n,m},\vartheta_{n,m})\cdot d(\theta,\vartheta).$ 
Using that $d(x,\theta)\leq d(x,\theta_{n,m})+d(\theta_{n,m},\theta)$ and $d(x,\theta_{n,m})\leq d(x,\theta)+d(\theta,\theta_{n,m}),$ we get that
$$R\leq\left\{\begin{array}{ll}-d(x,\theta)[d(\theta_{n,m},\vartheta_{n,m})-d(\theta,\vartheta)]+d(\theta,\vartheta)d(\theta_{n,m},\theta)
 & \mbox{if }D(x,P)\leq D_{\mbox{m}}(x,P)
\\
d(x,\theta_{n,m})[d(\theta_{n,m},\vartheta_{n,m})-d(\theta,\vartheta)]+d(\theta_{n,m},\vartheta_{n,m})d(\theta_{n,m},\theta) & \mbox{otherwise.}
\end{array}\right.$$
By (\ref{dt}) and (\ref{v}), we obtain then that $|D(x,P)-D_{\mbox{m}}(x,P)|\rightarrow_P 0.$
 \end{proof}

Next theorem 
states that the random depth is a  statistical functional depth, satisfying all six properties that constitutes the notion.
\begin{theorem}\label{Td} Let $(\mathfrak{F},d)$ be a functional metric space  and $P$ a randomly symmetric probability distribution on $\mathfrak{F}.$ 
Under the Assumptions  \ref{condV}  and  \ref{condD},
$D$ 
is a statistical functional depth. 
\end{theorem}
The proof of the properties P-1. to P-5.~follows from \cite{Nieto21a}, for $D$ and $D_r.$ 
The proof of property P-6. is in what follows. 
\begin{proof}
Given  $\{Q_n\}_n$ and $\{P_n\}_n,$ sequences of empirical distributions of $Q$ and $P$ respectively,
let us first prove that \begin{eqnarray}\label{P}\rho(Q,P)\rightarrow 0\end{eqnarray} implies $\rho(Q_n,P_n)\rightarrow 0,$ where $\rho(\cdot,\cdot)$ is the Prohorov metric.
 For that we make use of $Q_n\rightarrow Q$ and $P_n\rightarrow P,$ which by \citet[Theorem 11.3.3.]{rDudleyRAP} is equivalent to \begin{eqnarray}\label{nn}\rho(Q_n,Q)\rightarrow 0 \mbox{ and } \rho(P_n,P)\rightarrow 0.\end{eqnarray} Additionally, as $\rho(\cdot,\cdot)$  is a distance, we have that $\rho(Q_n,Q)=\rho(Q,Q_n).$
 Given $\epsilon>0$ and $A$ a Borel set, denoting $A^{\epsilon}:=\{y\in S:d(x,y)<\epsilon \mbox{ for some } x\in A\},$ we have that $(A^{\epsilon/3})^{\epsilon/3}= A^{2\epsilon/3}$ due to the following. Let $x\in A,$ $y\in A^{\epsilon/3}$ and $z\in (A^{\epsilon/3})^{\epsilon/3},$ then $d(x,y)<\epsilon/3$ and $d(y,z)<\epsilon/3.$ Therefore, by the triangular inequality, $d(x,z)\leq d(x,y)+d(y,z)<2\epsilon/3.$ Thus, for any $\epsilon>0$ and any $A$ Borel set, by the definition of  the Prohorov metric \cite[page 394]{rDudleyRAP}, (\ref{P}) and (\ref{nn}), $$P_n(A)\leq P(A ^{\epsilon/3})+\epsilon/3\leq Q(A^{2\epsilon/3})+2\epsilon/3\leq Q_n(A^\epsilon)+\epsilon$$ which implies $\rho(Q_n,P_n)\rightarrow 0.$
 
Taking $A_x=\mathfrak{H}^{\theta_{n,m}}_{x},$ with $\theta_{n,m}$ a random center of functional  symmetry of $P,$ we have that $P_n(A_x)\geq 1/2$ for all  $x\in S_m.$
 Then, as $\rho(Q_n,P_n)\rightarrow 0,$ we obtain $Q_n(A_x^{\epsilon})\geq 1/2-\epsilon.$ Thus, when $\epsilon\rightarrow 0,$ $Q_n(\mathfrak{H}^{\theta_{n,m}}_{x})\geq 1/2$ for any  $x\in S_m.$ 
 From here, following the ideas of the proof of above Theorem \ref{DMD} we obtain that $\lim_{Q\rightarrow P}|D(x,P)-D(x,Q)|=0.$
 \end{proof}

\begin{corollary}Given $(\mathfrak{F},d)$  a functional metric space satisfying Assumption \ref{condD} and $P$ a randomly symmetric probability distribution on $\mathfrak{F}$,
 $D_r$ 
satisfies properties P-2.~, P-3.~, P-4.~, P-5.~and P-6.~and does not satisfy property P-1.
\end{corollary}

The fulfillment and necessity of these results is studied in practice in the simulation Section \ref{Simu}.
Analogous results to those in this section are obtained for the notion of $\delta-$random symmetry.

\subsection{Proposed functional depth in the framework of  simultaneously continuous dimensions}\label{FMD}

The aim of this section is to provide an appropriate pseudo-distance for the random depth and the metric depth to be applied to a set of PET data in order to find the representative subject of the set. In this type of set, each datum (subject) consists on a functional datum in simultaneously four different continuous dimensions. 

Settings beyond  the one-dimensional functional case for which a notion of depth has been already proposed are the multivariate functional one and that of  data defined on a subset of $\mathbb{R}^2$ \citep{Genton2014}. In \cite{HubertC} and \cite{MultFuBD} the integrated data depth (Section \ref{sectionOdFD}) is extended to the space of multivariate ($\RR^k$) continuous functions on an interval $I\subset\mathbb{R},$ by substituting $D_1(x(t),P_t)$ in (\ref{FD}) by a multivariate depth: the Tukey depth, the random Tukey depth, when $k$ is large, or the simplicial depth.
 Note that the setting in \cite{HubertC} and \cite{MultFuBD} considers a space with two variables one continuous (the time) and one atomic with $k$ elements in the support. This is far from our actual setting as we consider four continuous variables, three spatial, $X, Y$ and $Z,$ and the time. Nevertheless, the integrated data depth is easily extended to our continuous setting by simply substituting  $I\subset\mathbb{R}$ in (\ref{FD}) by  $I\subset\mathbb{R}^k,$ with $k=4$ in our case. In fact, a particular case of this,  for $k=2,$ is the proposal in \citet{Genton2014}. 
 However, there is still a remaining issue due to the integrated data depth involving the computation of $D_{1}(x(t),P_t)$ at each grid point $t, $ which requires the values of the different PET scans to  be comparable at  each fixed grid point. This is not the case for this type of data (see Section \ref{ARD}), unless they have been preprocessed to lie on a common template, which can introduce other artifacts into the analysis.

 Having said that, in Definition \ref{HD} we propose a pseudo-distance for the random depth and the metric depth to be applied to data which are theoretically measured in different continuous directions and which allows different grids for each subject, in practice.

\begin{definition}\label{HD}
Let $\mathfrak{F}$ be a functional space over  $I\subset\RR^k$ with $\lambda_k(I)\in(0,\infty)$ ($\lambda_k$ the Lebesgue measure in $\RR^k$) and such that  $x(t)\in\RR$  for every $t\in I$ and any $x\in\mathfrak{F}.$ For any $x,y\in\mathfrak{F}$   integrable on $I,$ the \emph{hyperbolic function} is defined as
$d(x,y):=|i(x)-i(y)|,$  where
$i:  \mathfrak{F}  \longrightarrow   E$ with  $$i(x):=
 \displaystyle\frac{\displaystyle\int_{I}x(t)dt}{\displaystyle\int_{I}I_x(t)dt}, \mbox{ } \mbox{ if } \lambda_k(x\neq 0)>0$$ and $i(x):=0$ otherwise.
There, $E\subset \RR$ is the codomain of $i$ and  $I_x(t)$ takes value 1 if $x(t)\neq 0$ and 0 otherwise.
\end{definition}

While this definition is useful for our particular application below, it can easily be generalised to any setting. It will take positive values for those elements in the grid where the distance is relevant and zero for cases where it is irrelevant.

\begin{proposition}\label{dist}
 The hyperbolic function is a pseudo-distance.
\end{proposition}
\begin{proof}
$d(x,y)$ coincides with the, Euclidean, distance between the real numbers $i_x$ and $i_y,$ from which the result is derived.
\end{proof}

\section{Comparison with existing one-dimensional depths for big functional data}\label{sectionOdFD}
 As our focus is on
the
type of data involved in neuroimaging deconvolution, we restrict ourselves to one-dimensional functional depths that are computationally effective when handling big data and that provide deepest value to a generalised median. Among the existing notions, the ones that fall in this category are the family of notions known as integrated data depth \citep{FraimanMuniz}, which includes the modified band depth in the particular case of $J=2$ \citep{Romo09, Genton}. Also in this category are the band depth with $J=2$ \citep{Romo09, Genton} and the random Tukey depth \citep{randomTukey}. This section will discuss these examples of depth but with special attention to the integrated data depth, as in practice, this depth seems to be the most competitive amongst its peers. Besides these definitions of depth, a variety of definitions of one-dimensional functional depth have been provided in the literature. For instance, the spatial depth \citep{ChakrabortyAoS2014}; see \citet{NRBattey2015} for a review.


Given a function $x=x(t)$ in a one-dimensional functional space, usually the space of continuous functions on a compact interval $I\subset\mathbb{R},$ and a distribution $P$ on the space,  the \emph{integrated data depth} of $x$ with respect to $P$ is
\begin{eqnarray}\label{FD} ID(x,P):=\int_{I\subset\mathbb{R}}D_1(x(t),P_t)dt, \end{eqnarray}
 where $P_t$ denotes the marginal of $P$ at time $t.$  
There are two natural ways of computing a depth in $\RR:$  the Tukey depth \citep{Tukey} \begin{eqnarray}\label{TD} D_{1}(x(t),P_t)=\min(P_t(-\infty, x(t)],P_t[x(t),\infty))\end{eqnarray}  and the simplicial depth \citep{Liu90} \begin{eqnarray}\label{SD}D_{1}(x(t),P_t)=P_t(-\infty, x(t)]\cdot P_t[x(t),\infty).\end{eqnarray}
We denote $ID$ in (\ref{FD}) by $D_I$ when using (\ref{TD}) and by $D_M$  when using (\ref{SD}), as the empirical integrated data depth  and modified band depth with $J=2$ coincide when $D_1$ is taken as the simplicial depth on $\mathbb{R}$ \citep{HubertC}.
In \cite{FraimanMuniz} it is also proposed to use $D_{1}(x(t),P_t)=1-|.5-P_t(\infty,x(t)]|,$ but we do not use it here because of its performance when applied to discrete distributions. For instance, given the distribution that applies equal probability to the points $x_1=-1,$  $x_2=-.5,$  $x_3=0,$  $x_4=.5$ and  $x_5=1,$ it is clear that $x_3=0$ is the only deepest element  among the $x_i$ for $i=1, \ldots, 5,$ but this depth associates deepest value also to $x_2.$

The band depth, $D_B,$ was introduced in \cite{Romo09} as a one-dimensional functional depth for samples where the curves rarely cross over, and it was generalized to the modified band depth for when those crossovers exist. Due to the nature of the one-dimensional curves associated to the deconvolution of PET data (see Figure \ref{Oc}) it is not recommended to use the band depth function, because of the crossovers. The random Tukey depth, $D_{T},$ involves the projection of the data in a set of curves drawn with a particular distribution. The use of the random Tukey depth as a multivariate depth has been well studied, however, as a functional depth \citep{LibroMieres} it is not completely clear how to select the mentioned  distribution. Thus, we do not recommend basing methodology for deconvolution of dynamic neuroimaging data on random Tukey depth functions. Nonetheless, we will examine the practicality of all these depths in the empirical applications.  With respect to their theoretical properties, according to  \cite{NRBattey2015},  $D_M, D_B$ and $D_{T}$ satisfy P-1., P-2G.~and P-4.~and do not satisfy P-3.~and P-5.
P-6.~is satisfied  by $D_{T}$ when the limiting distribution is continuous or the sequence of distributions is the sequence of empirical distributions; by $D_B$ when $\mathfrak{F}$ is restricted to be the space of equicontinuous functions on $\mathcal{I}\subset \RR$; and by $D_M$ in general.
%
The result for $D_I$ is below. Its proof follows easily from existing results in the literature
\begin{theorem}\label{TI}
Let $(\mathfrak{F},d)$ be the space of continuous functions on $I$ with the distance induced by the supremum norm.
$D_I$ satisfies P-1., P-2G.~and P-4.~and does not satisfy P-3.~and P-5.
P-6.~is satisfied  by $D_I$  when the limiting distribution is continuous or the sequence of distributions is the sequence of empirical distributions.
\end{theorem}
\begin{proof}

P-1. It derives from the invariance of the Tukey depth as $D_1(f(x),P_{f(X)})=D_1(x,P_X).$ 

P-2G. As $P$ is a zero-mean, stationary, almost surely continuous Gaussian process on $I,$ $P_t$  is  a zero-mean Gaussian distribution in $\RR$ for any $t\in I.$ Then, $D_1(0,P_t)=\sup_{x\in\RR}D_1(x,P_t)$ for all $t\in I$ and, consequently, $\sup_{x\in\mathfrak{F}}\int_ID_1(x(t),P_t)dt=D(0,P).$

P-3. The proof is by counterexample. Let $P$ be a discrete distribution with $P\{x_1\}=P\{x_2\}=1/2$ and such that $x_1(t)>x_2(t)$ for all $t\in I.$ Then, there exist infinitely many elements with maximum depth in the space of continuous functions; which contradicts \citet[Lemma 3.5]{NRBattey2015} that states that a depth that satisfies P-3 is uniquely maximized.

P-4.  As the cumulative distribution function,  as a function of $x,$ is upper semi-continuous, for all $\epsilon_0>0$ there exists a $\delta_0>0$ such that $D_1(y,P)\leq D_1(x,P)+\epsilon_0$ for all $x,y\in\RR$ with $|x-y|<\delta_0.$ Thus, given $\epsilon>0$ we take $\epsilon_0=\epsilon/\lambda(I)$ and $\delta=\delta_0$ and therefore for any $t\in I$, \begin{eqnarray}\label{e4} D_1(y(t),P_t)\leq D_1(x(t),P_t)+\epsilon_0 \end{eqnarray} for all $x,y\in\mathfrak{F}$ with $|x(t)-y(t)|<\delta.$   Consequently, taking the integral over $t\in I$ on equation (\ref{e4}), we obtain that
$D(y,P)\leq D(x,P)+\epsilon_0\lambda(I)$ for all $x,y\in\mathfrak{F}$ with $d(x,y)<\delta.$

P-5.
The proof is by counterexample. We follow a counterexample used in the proof of  \citet[Theorem 4.7]{Nieto21a} but state it here for the sake of completeness. Let $P$ be a discrete probability with $P[x_i]=1/3$ for $i=1,2,3$ and $x_1(t)>0,$  $x_2(t)=0$ and  $x_3(t)<0$ for all $t\in I,$ with $x_1$ and $x_3$ non-constant functions. As $\alpha(t)>0$ for all $t\in I,$ $f(x_1)(t)>0,$ $f(x_2)(t)=0$ and $f(x_3)(t)<0$ for all $t\in I.$ As the transformation simply shrinks the convex hull over the $L_{\delta}$ region, whilst the one dimensional Tukey depth at any time point computed  on the transformation is the same as before. It is thus immediate that $D(x_1,P_X)=D(f(x_1),P_{f(X)}).$

P-6.
It follows from the proof of the random Tukey depth in Theorem 4.8 of \cite{Nieto21a} by using that $P$ continuous implies that $P_t$ is continuous for all $t\in I$ and that $\lambda(I)<\infty.$
\end{proof}



\section{Simulations} \label{Simu}
This section illustrates through simulations the theoretical properties of the proposed depth. It suffices for that to analyze one-dimensional functional curves. Equivalent results are obtained, though, when simulating curves that are simultaneously continuous in various dimensions. 

To simulate the type of curves  associated with the deconvolution of PET data (Figure \ref{Oc}), we define a sequence of functional random variables  $\{X_{k,N}\}_{k=1, \ldots,N},$ where for all $t\in (0,1),$
\begin{eqnarray}\label{X}
X_{k,N}(t):= & x_{k,N}(t)+Y_k\cdot Z_k(t),\\
\label{x}
 x_{k,N}(t):= & \displaystyle\exp[-Ntc/(k+N/c)].
\end{eqnarray}
We take  $c$ to be a real positive constant, $Y_1, \ldots, Y_N$ a simple random sampling following a Bernoulli distribution of parameter 1/3 and  for each $t\in(0,1),$ $Z_1(t), \ldots, Z_N(t)$ a simple random sampling following a mean zero normal distribution with standard deviation $sd.$ For practical purposes, we assume that the functions $\{X_{k,N}\}_{k=1, \ldots,N},$ and $\{x_{k,N}\}_{k=1, \ldots,N},$  are recorded at the set of discretization points
\begin{eqnarray}\label{I}I:=\{i\cdot10^{-1}: i=1, \ldots, 99\}\subset(0,1).\end{eqnarray}

The functional random variables  $\{X_{k,N}\}_{k=1, \ldots,N}$ are the result of applying a small perturbation, or error, to the non-random functions  $\{x_{k,N}\}_{k=1, \ldots,N}.$ 
The reason for the $Y_k$'s, $k=1, \ldots, N,$ to follow a Bernoulli distribution with parameter 1/3 is that then, only some elements of the sample are affected by the white noise error, introduced by $Z_1, \ldots, Z_N.$ It generally occurs to real data that the errors does not affect all the elements of the sample, or not all of them in the same measure, and therefore the introduction of the Bernoulli distribution, while obviously not a true representation of the system, produces simulated data closer to real data. Another possibility is to build randomness in $sd,$ i.e. to draw $sd_1, \ldots, sd_N$ from some real distribution and use $sd_k$ as the standard deviation of each $Z_k(t),$ $k=1, \ldots, N.$

\begin{figure}[htb]
\begin{center}
\includegraphics[height=5cm,width=.3\linewidth]{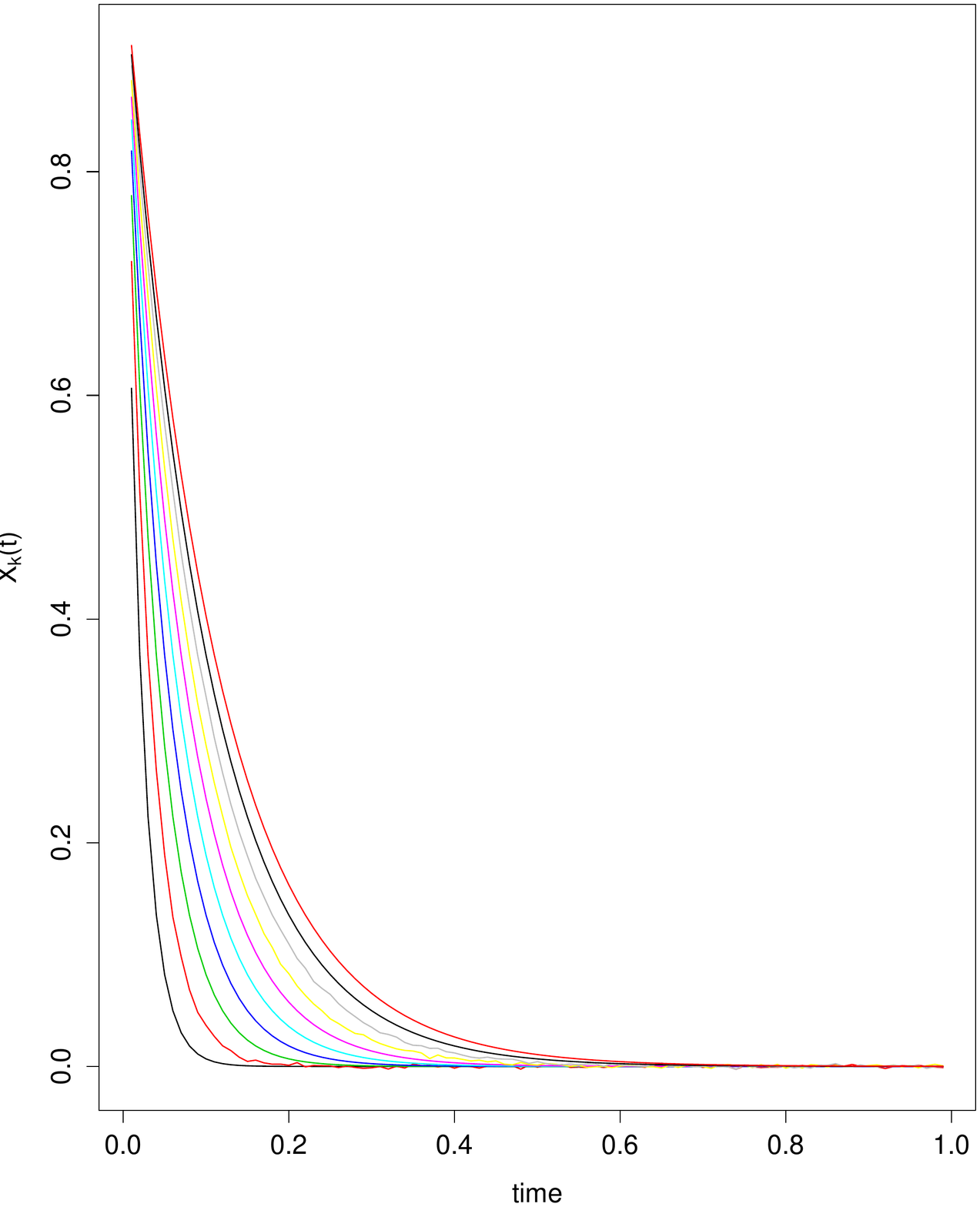}
\includegraphics[height=5cm,width=.3\linewidth]{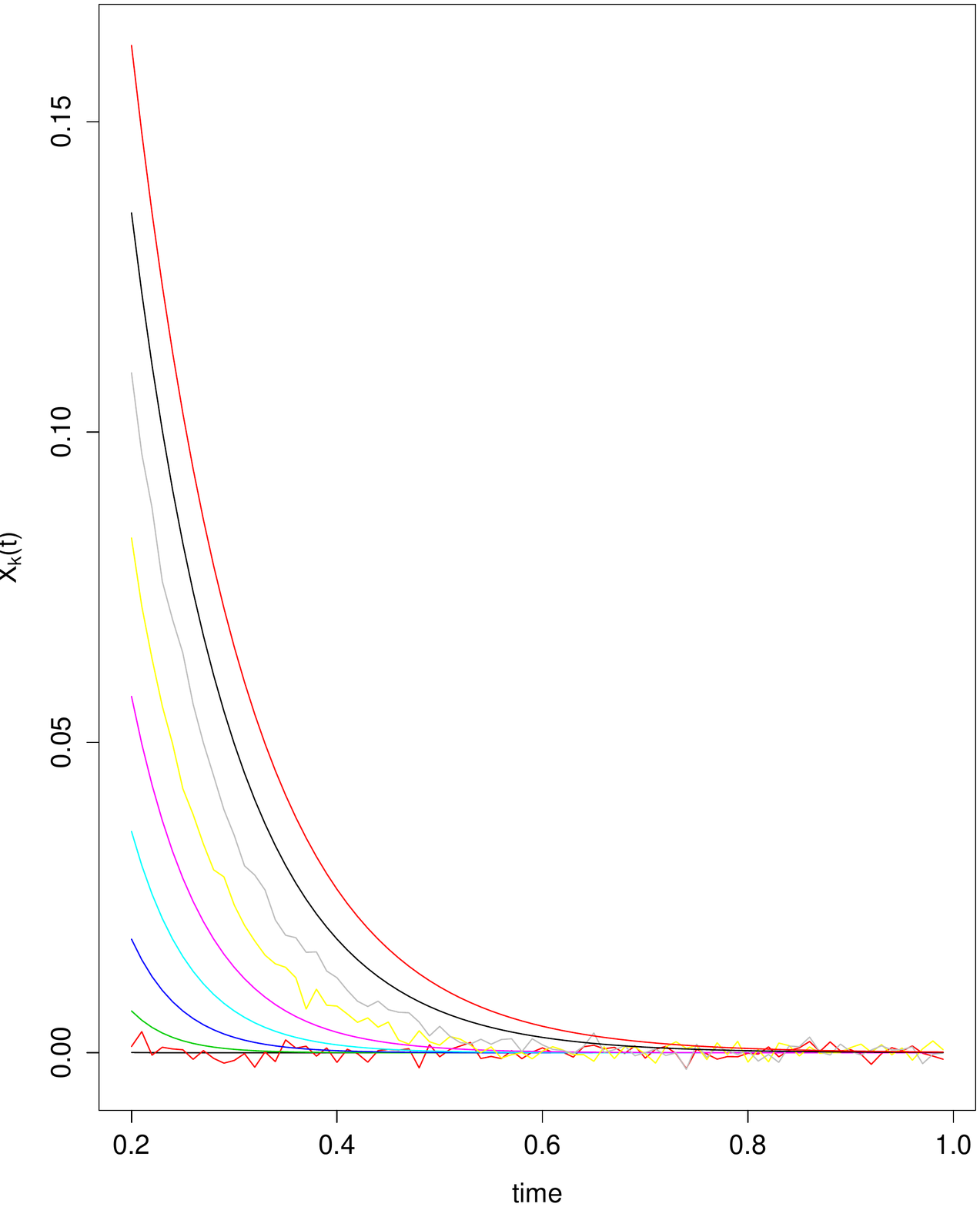}
\includegraphics[height=5cm,width=.3\linewidth]{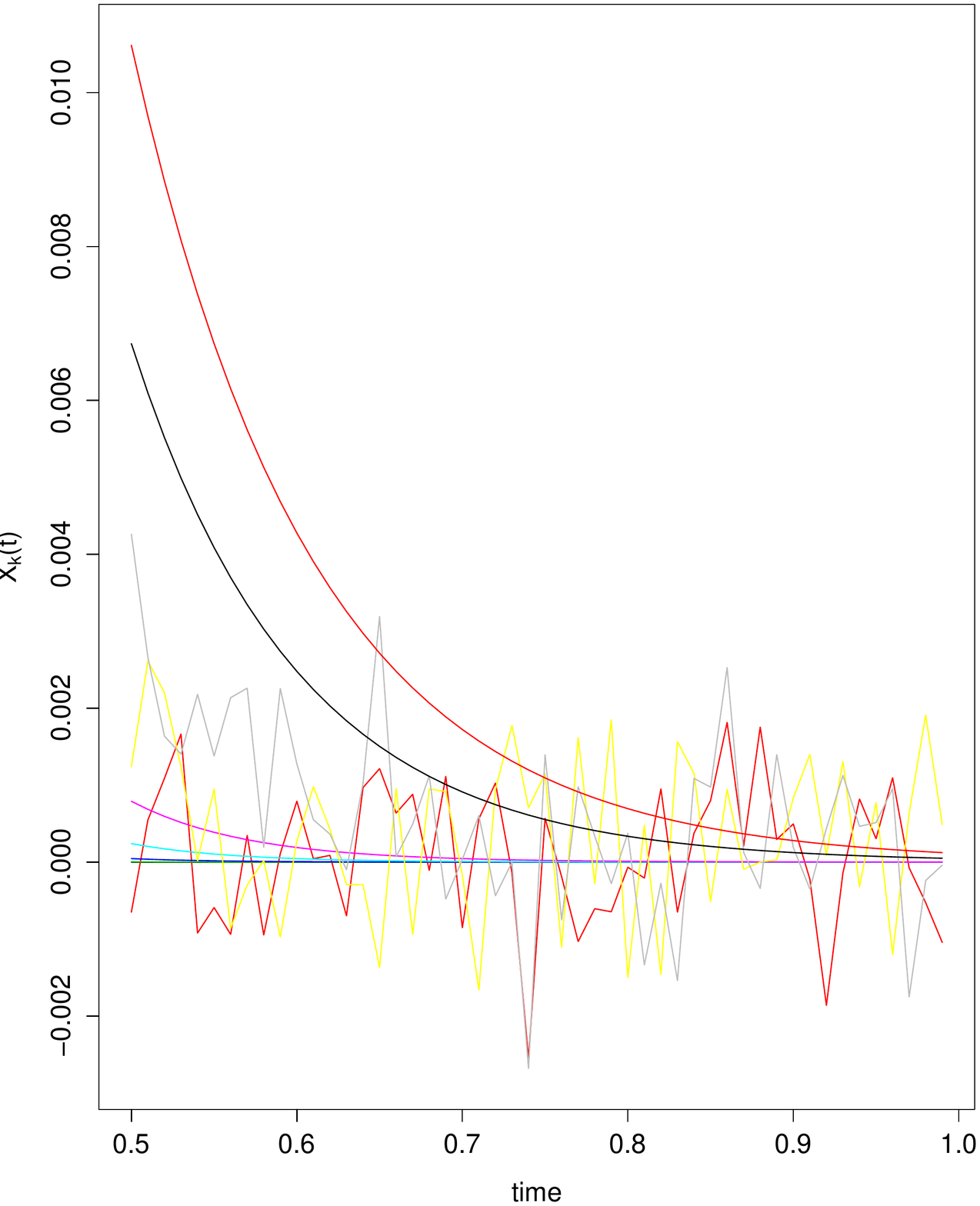}
\\
\includegraphics[height=5cm,width=.3\linewidth]{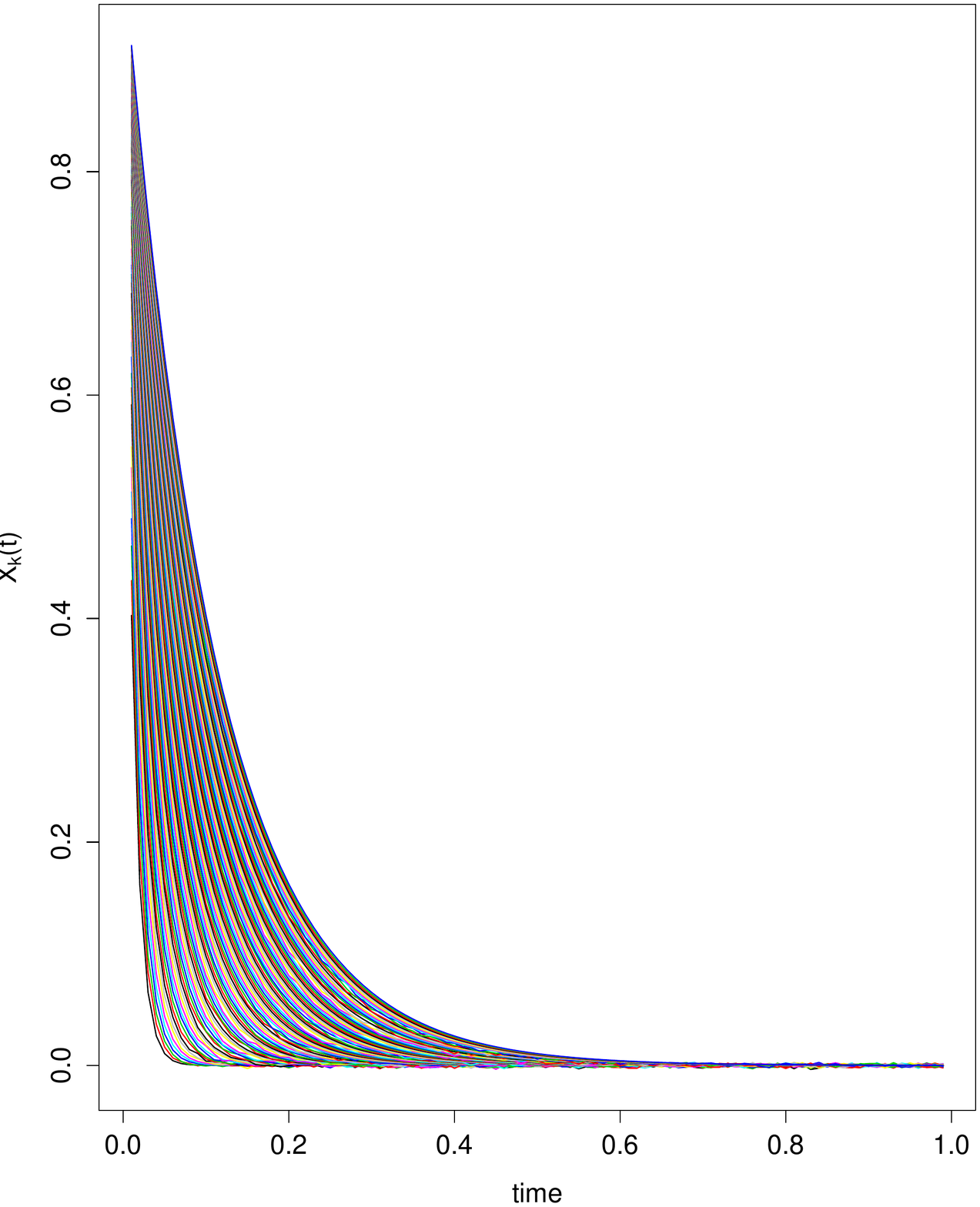}
\includegraphics[height=5cm,width=.3\linewidth]{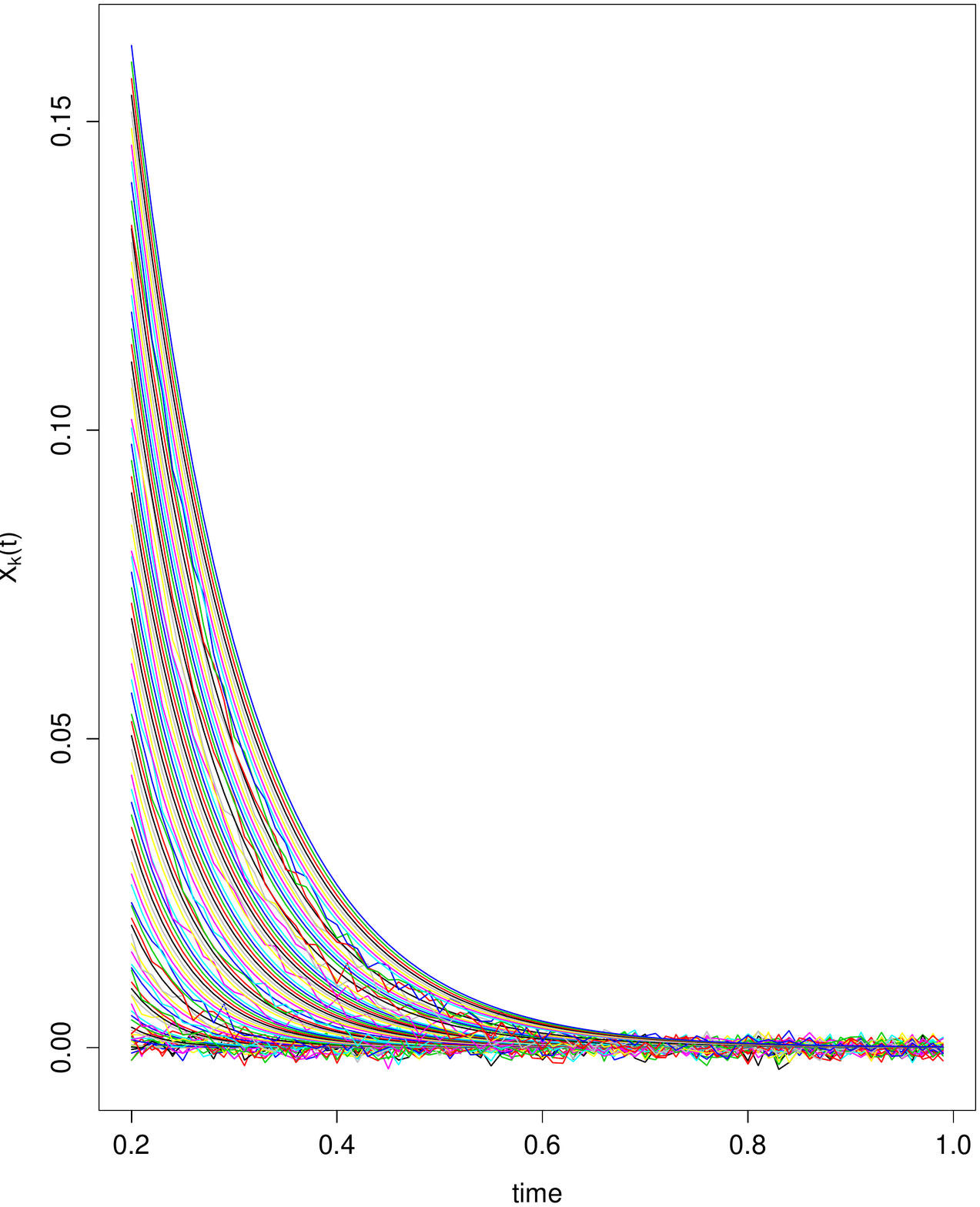}
\includegraphics[height=5cm,width=.3\linewidth]{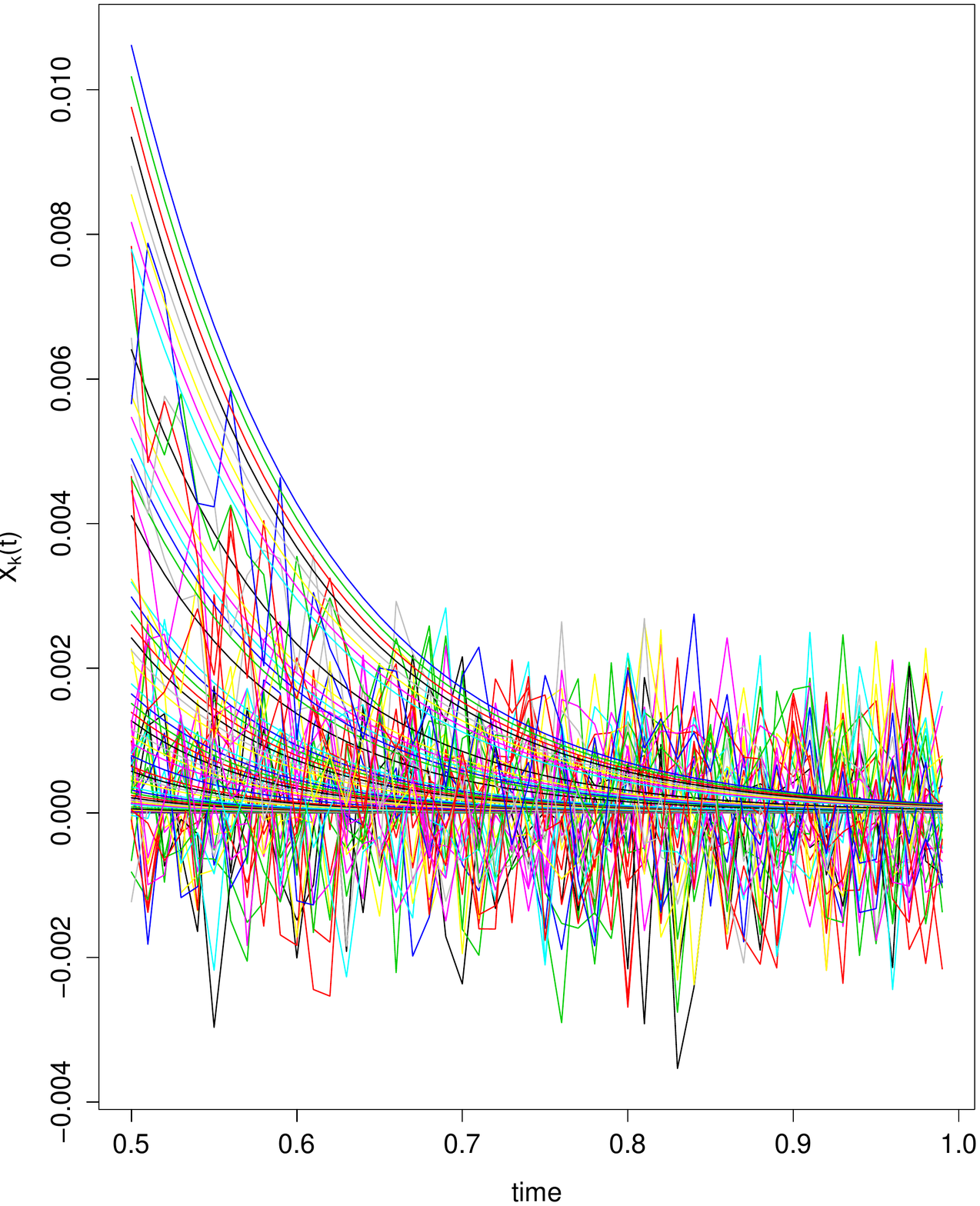}
\end{center}
\caption{Realization of $\{X_{k,N}\}_{k=1, \ldots,N}$ with $sd=.001$  and $c=10$ for $N=10$ (top row) and $N=100$ (bottom row); plotting the image of the curves on the whole domain (left column), on the domain $[.2,1)$ (middle column) and in the domain $[.5,1)$ (right column). }
\label{Si}
\end{figure}

To illustrate these random variables, in the left column of Figure \ref{Si} we have plotted a realization of $\{X_{k,N}\}_{k=1, \ldots,N}$   for $N=10$ (top row) and $N=100$ (bottom row). The panels in the central column of the figure correspond to a zoom on the domain $[.2,1)$ of the plots in the left column and in the right column to a subsequent zoom on the domain $[.5,1).$  $sd=.001$ is taken as the standard deviation and $c=10$ as the real positive constant. That  only some of the curves in the figure are perturbed (approximately 1/3 of them) is due to the use of the Bernoulli distribution of parameter 1/3.

\subsection{Convergence of the random depth}

This section investigates, through simulations, the random depth in a \textit{big} functional data framework. It relates to Theorems \ref{td} and \ref{DMD}. Using the random depth, we compute the  distance between a random deepest element of $\{x_{k,N}\}_{k=1, \ldots,N},$ defined by (\ref{x}), and the set of true deepest elements. In Table \ref{t1}, the mean and standard deviation of these distances are recorded, when a random deepest element is computed 100 times. This is done for $N=250000,$ $c=10^{-3}$ and $n,m\in\{10,100,1000\}.$ The selection of $N$ is based on the population size we encounter in the real data set used in Section \ref{sectionAD}. Here, we use the distance associated to the space of Lebesgue square integrable functions, generally referred to as the $\mathbb{L}_2$ distance.
For the simulations, the $\mathbb{L}_2$ distance is not exact, as we estimate it using the set of discretization points $I,$ defined by (\ref{I}), at which we assume the data are recorded.

\begin{table}[!htb]
\caption{Mean and standard deviation over 100 repetitions of the distance between a  deepest curve, among $\{x_{k,N}\}_{k=1, \ldots,N}$ computed with the random depth with $n,m\in\{10,100,1000\},$ and the set of true deepest curves. $N=250000$ and $c=10^{-3}.$}
\begin{center}
\begin{tabular}{r|ccc}\label{t1}
&$m=10$&$m=100$&$m=1000$
\\
\hline
\multirow{2}{*}{$n=10$}&$2.471075\cdot 10^{-8}$&$1.065400 \cdot  10^{-8}$&$1.623565\cdot 10^{-8}$
\\
&$(7.247131\cdot 10^{-8})$&$( 4.828071 \cdot  10^{-8})$&$(7.211914\cdot 10^{-8})$
\\ \\
\multirow{2}{*}{$n=100$}&$9.332167\cdot 10^{-10}$&$4.457600 \cdot  10^{-10}$&$9.390100 \cdot  10^{-10}$
\\
&$(1.923750\cdot 10^{-9} )$&$( 7.697200 \cdot  10^{-10})$&$(3.558400 \cdot  10^{-9} )$
\\ \\
\multirow{2}{*}{$n=1000$}&$1.013214\cdot 10^{-10}$&$1.052271 \cdot  10^{-10}$&$1.142058\cdot 10^{-10}$
\\
&$(1.011335\cdot 10^{-10})$&$(9.209343 \cdot  10^{-11})$&$( 1.175811\cdot 10^{-10})$
\\
\hline
\end{tabular}
\end{center}
\end{table}

From Table \ref{t1}, it can be seen that for each fixed $m$, the recorded means and standard deviations slightly decreases when $n$ increases. For each fixed $n,$ however, they    are stable and already low for $m=10.$ To study its stability, Table \ref{t2} displays for $m=N,$ and $n\in\{10,100,1000\},$ the mean and standard deviation over 100 distances between the random and true deepest elements. From Table \ref{t2} it is observed that the stability remains; with more stability as $n$ increases. In fact, for $n=1000,$ the obtained mean value is even smaller when $m=1000$ (Table \ref{t1}) than  when $m=N$  (Table \ref{t2}).

\begin{table}[!htb]
\caption{Mean and standard deviation over 100 repetitions of the distance between a  deepest curve, among $\{x_{k,N}\}_{k=1, \ldots,N}$ computed with the random depth with $n\in\{10,100,1000\}$ and $m=N,$ and the true deepest curves. $N=250000$ and $c=10^{-3}.$}
\begin{center}
\begin{tabular}{r|cccl}\label{t2}
&n=10&n=100&n=1000
\\
\hline
mean &$6.143400 \cdot  10^{-9}$&$ 6.505600 \cdot  10^{-10}$&$1.303000\cdot  10^{-10}$
\\
stad. deviation &$2.937200 \cdot  10^{-8}$ &$1.809600 \cdot  10^{-9}$&$1.455767\cdot  10^{-10}$
\\
\hline
\end{tabular}
\end{center}
\end{table}

Furthermore, in order to understand how close to zero are the values displayed in Tables \ref{t1} and \ref{t2}, we have computed the estimated $\mathbb{L}_2$ distance based on $I$ of the two deepest elements among $\{x_{k,N}\}_{k=1, \ldots,N},$ obtaining $3.418716943\cdot 10^{-14}.$ 

\subsection{Importance of the properties of the notion of depth}
This section shows through simulations the importance of a depth satisfying the properties constituting the notion of functional depth. It relates to Theorem \ref{Td}. According to Theorem \ref{TI}, the existing functional depths
$D_I, D_M, D_B$ and $D_{T}$ fail properties P3.~and P-5. As commented in \cite{NRBattey2015}, P-3.~is useful for classification problems, which is not the focus of this paper. However, P-5.~is of great importance here as an aim of this paper is to compute the deepest  curve among a set of real data curves of the type shown in Figure \ref{Oc}. This type of curves has the particular characteristic of containing a region of the domain in which the curves exhibit little variability in comparison with the rest of the domain. A characteristic of real data it that it contains errors, which is made very apparent in the regions with little variability.
Thus, if P-5.~is not satisfied, the depth can be heavily influenced by the errors in the regions with little variability.

The aim of these simulations is to practically determine which existing measures of functional depth are affected by the error in the data when selecting the deepest element. With this aim in mind, we first study the non-random case, which contains no error in the depth.
We concentrate on the deepest element due to the major role it will play in Section \ref{sectionAD} when computing the deconvolution of neuroimaging data. Despite the fact that when we are using the random depth, the deepest element is a random center of symmetry, we refer to it in this section as the deepest element and not the center to be consistent with the nomenclature when using the other existing examples of depth.
It is algebraically obvious that the deepest elements amongst $\{x_{k,N}\}_{k=1, \ldots,N}$ are $x_{N/2,N}$ and $x_{N/2+1,N}$ when $N$ is even and $x_{(N+1)/2,N}$ when $N$ is odd. The simulations in this section run for samples of moderate size on the range from 10 to 1001. There is no need to use large sample sizes as the non-fulfillment of property P-5. is observable independently of the sample size. 
Additionally, here, we use both even and odd sample sizes as there can be major difficulty in depicting the proper deepest elements when $N$ is even; as there are two elements to depict as opposed to the one when there are odd sample sizes. As the sample sizes used are moderate, our proposal just uses as a deepest element the center of $d$-symmetry.

In the non-random case (i.e. when the $sd$ of the noise is zero) we obtain the deepest element, in practice, for all the different examples of depth function except for the random Tukey depth. This can be seen in the top panel of Table \ref{t0} which displays the proportion of times, out of 100, that the different examples of depth depict the deepest element(s) of $\{x_{k,N}\}_{k=1, \ldots,N}.$ The random Tukey depth is computed by randomly projecting the functions $\{x_{k,N}\}_{k=1, \ldots,N}$. Thus, in doing so, we do not necessarily  obtain the right deepest element(s) because it is possible to obtain that $\int_{[0,1]}x_{k,N}(t)\alpha(t)dt>\int_{[0,1]}x_{k+1,N}(t)\alpha(t)dt$
despite $x_{k,N}(t)<x_{k+1,N}(t)$ for all $t\in(0,1),$ as if $\alpha(t)<0,$ we have that $x_{k,N}(t)\alpha(t)>x_{k+1,N}(t)\alpha(t).$ This is of course aggravated when the number of projections used increases from 1 to 10. Thus, the values obtained in the top panel of Table \ref{t0} are worse when the random Tukey depth is based on 10 random projections than when it is based on a single one projection.

\begin{table}[!htb]
\caption{Proportion of times, out of 100, that the different examples of depth depict the appropriate deepest element among $\{x_{k,N}\}_{k=1, \ldots,N}$ when the observed data is drawn from $\{X_{k,N}\}_{k=1, \ldots,N},$ defined by (\ref{X}). We take $c=10$ and $sd=0$ (top panel), $sd = 10^{-4}$ (middle panel) and $sd = 10^{-3}$ (bottom panel).
This is done for a variety of even and odd sample sizes in the range 10 - 1001. The depth used are the integrated data depth ($D_I$), modified band depth ($D_M$), band depth ($D_B$), random Tukey depth with 1 projection  ($D_{1}$) and 10 projections ( $D_{10}$) and the random (metric) depth ($D$).}
\begin{center}\label{t0}
\begin{tabular}{r|cccc||cccc}
\multicolumn{1}{r}{}&\multicolumn{8}{c}{standard deviation $sd = 0$}\\
&\multicolumn{4}{c}{$N$ even}&\multicolumn{4}{c}{$N$ odd}\\
&N=10&N=100&N=500&N=1000&N=11&N=101&N=501&1001\\
\hline
$D_I$ & 1 &1 & 1& 1&1 & 1 & 1&1
\\
$D_M$& 1 &1 & 1& 1&1 & 1 & 1&1
\\
$D_B$&1 &1 & 1& 1&1 & 1 & 1&1
\\
 $D_{1}$& .80 &.67 &.69 &.73 &.78 & .76 &.74 &.76
\\
$D_{10}$& .11 &.08 &.05 &.07 &.12 &.05  &.06 &.04
 \\
D& 1 &1 & 1& 1&1 & 1 & 1&1
\\
\hline
\multicolumn{1}{r}{}&\multicolumn{8}{c}{}\\
\multicolumn{1}{r}{}&\multicolumn{8}{c}{standard deviation $sd = 10^{-4}$}\\
&\multicolumn{4}{c}{$N$ even}&\multicolumn{4}{c}{$N$ odd}
\\
&N=10&N=100&N=500&N=1000&N=11&N=101&N=501&1001\\
\hline
$D_I$ & .05 &.00 &.00 &.04 &.75 &.45  &.20 &.00
\\
$D_M$& .07& .00& .00&.00 &.64 &.09  &.00 &.00
\\
$D_B$&.17 &.00 &.00 &.02 &.23 &.02  &.02 &.00
\\
  $D_{1}$& .80 &.71 & .73&.66 &.78 &.73  &.70 &.59
\\
$D_{10}$& .09 &.04 &.02 &.03 &.18 &.02  &.05 &.00
 \\
$D$&1 &1 &1 &1 &1 & 1 &1 &1
\\
\hline
\multicolumn{1}{r}{}&\multicolumn{8}{c}{}\\
\multicolumn{1}{r}{}&\multicolumn{8}{c}{standard deviation $sd = 10^{-3}$}\\
&\multicolumn{4}{c}{$N$ even}&\multicolumn{4}{c}{$N$ odd}
\\
&N=10&N=100&N=500&N=1000&N=11&N=101&N=501&1001\\
\hline
$D_I$ & .01 &.01 &.00 &.00 &.68 &.41  &.12 &.03
\\
$D_M$&.01 &.00 &.00 &.00 &.55 &.12  &.00 &.00
\\
$D_B$&.19 &.00 &.00 &.00 &.21 &.00  &.02 &.00
\\
$D_{1}$& .78 &.68 &.33 &.20 & .84& .68 &.41 &.16
\\
$D_{10}$& .11 &.04 &.00 &.00 &.15 &.01  &.00 &.01
 \\
$D$&1 &1 &.41 &.54 &1 &1  &.76 &.72
\\
\hline
\end{tabular}
\end{center}
\end{table}

In the middle and bottom panel of Table \ref{t0} we study how well, in practice, the different examples of functional depth satisfy property P-5. Thus, in the last two panels  of Table \ref{t0} it is recorded the proportion of times, out of 100, that the different examples of depth depict the deepest element(s) of $\{x_{k,N}\}_{k=1, \ldots,N}$ when the observed functions are those simulated with $\{X_{k,N}\}_{k=1, \ldots,N}$  for $sd = 10^{-4}$ (middle panel) and $sd = 10^{-3}$ (bottom panel). A $sd=10^{-3}$ is small in a neuroimaging study. There is no need, however, to increase it because already  the success rate of the existing examples of depth decreases in general abruptly, with the exception of the random (metric) depth and the random  Tukey depth with 1 projection.
Indeed, the only success rates larger than $.2$ are: (i) those that have sample size 11 when not being computed with the random Tukey depth using 10 projections and, (ii) those  that have sample size 101 when computed with the integrated data depth ($D_I$). The results of the random Tukey depth with 1 projection are in general equivalent to those computed when $sd=0.$ In fact, in Table \ref{t0} there is only a visible decrease of this depth as the sample size and $sd$ simultaneously increases. The success rate for the random (metric) depth is always larger than those obtained with any of the other examples of depth. In fact, it is always 1 but for $sd = 10^{-3}$ when the sample sizes are larger than 101.

\section{Applications to neuroimaging}\label{ARD}

An observed PET datum consists of measures of radioactivity concentration on a 3-dimensional (3D) spatial grid over a time grid and, in practice, it can be viewed as a 4D array. The observed PET data, using the PET tracer [$^{11}$C]-Diprenorphine, we analyse in this paper consists of 4D arrays of dimension 128, 128, 95 and 32, and originally comes from the study into opioid receptor density by \cite{hammers:2007}.  The data has been subsequently analyzed by \cite{jiang2009smoothing,FPCA}.
Generally, the first three dimensions are respectively referred to as the dimensions of $X, Y$ and $Z$  and the fourth as the time. Taking this into account, a PET datum consists of a set of $128\cdot 128\cdot 95$ curves over time, which reduces to $N$ curves over time after applying a mask (see Section \ref{sectionAD} for information on the mask and $N$).

Deconvolution of PET data is an important step of its analysis, as the original curve is contaminated by the input function, which is subject specific and usually unrelated to the quantities of interest. Previous procedures (see \cite{Goldsmith, FPCA} 
and the references therein) are based on the selection of certain basis functions so are nonparametric but not in the sense of a depth based approach.

In order to provide a truly non-parametric (depth based) deconvolution of a PET datum, we propose to first obtain the depth of the $N$ given curves and select a deepest curve, according to the random depth 
(see Section \ref{sectionFD}), due to $N$ being large in practice.
We then provide an image based measure associated with the deconvolution based on the depth. In the image, we ignore elements outside the brain and set to $0$ voxels corresponding to the curves which although inside the brain have an integral over time (also known as intensity) less than the one of the deepest curve. This is because we are intrinsically interested in areas where there is high uptake of the radiotracer. The rest of the voxels in the image are then set to the corresponding depth value for the curve at that location. 
This, in effect, orders the image locations and sets voxels with similar curves to similar depth values (allowing us to compare across subjects, as it is now the depth values that are being compared). We study the application of the random depth to the deconvolution of a real PET datum in Section \ref{sectionAD} .

Given a set of PET data we also propose to apply the notion of depth to determine a representative of the set and the distance that the rest of the elements have from the chosen element. We obtain the representative brain image by simply selecting the deepest element of the set, according to the metric depth when using the distance of simultaneously continuous dimensions proposed in Section \ref{FMD}. In computing the depth of the elements in the set, we order them and so provide a degree of representability of each subject in the dataset. It suffices to use the metric depth, as opposed to the random depth, as the number of subjects is generally very small, particularly compared to the number of voxels; for example, there are nine subjects in the dataset studied in Section \ref{sectRep}. The difficulty in the selection of the appropriate depth function here resides, though, in each subject being measured on a different grid, which can only be accounted for with either the metric or the random depth.
In the field of functional data analysis, it is almost always the case that different replications are measured on the same grid before applying any particular analysis. As this is far from the case here, the established analyses would pass the data through a preprocessing step such as interpolation, smoothing and/or the use of warping functions. However, our aim is to find the deepest subject without preprocessing procedures, as the representative subject is itself a useful tool in choosing the preprocessing required by the data. 
\subsection{Application to deconvolution of dynamic neuroimaging data}\label{sectionAD}

In this section, we consider the application of functional depth techniques to the deconvolution of PET image data. The PET image comes from the study in \cite{hammers:2007}, concerning the role of opioid (pain) receptors in normals and epilepsy patients. These images have been previously used in \cite{FPCA, jiang2009smoothing} and in this section we compare our proposed methodology with the results obtained in those papers and with the standard approach taken in PET, known as Spectral Analysis \citep{CunnJ:93,Gunn01}. The objective here is to obtain results similar to the ones given by the parametric approach but without the use of  any parametric assumption.

First, we regard a subject whose image is labeled as 2913 (the labelling refers to the study number when the data was collected). Although the PET image of the subject under study is reconstructed as a cuboidal array,   
our interest is only in those voxels that are inside the brain. Thus, to eliminate the background, we mask the data. Here, to produce the mask, we choose a threshold (20,000) that (on the images summed over time) leaves the brain intact yet eliminates most of the background of the image. This is a conservative value with which we choose the voxels of interest plus some of the background noise voxels, and in total, in the image under study, $N=254,807$ brain voxels remain out of $1,556,480$ voxels in the original array. 

We consider the $N$ selected voxels over time as $N$ one-dimensional functional data. In the left plot of Figure \ref{Oc} (Section \ref{sectionIn}) we have  plotted, in a variety of colors,  $100$ out of these $N$ curves selected at random. In the right plot, the same $100$ curves are plotted in grey, while the deepest curves among the $N$ according to the different depth methods commented in previous sections are also shown; red for the integrated data depth ($D_I$), black for the modified band depth ($D_M$), green for the band depth ($D_B$), orange for the random Tukey depth ($D_{T}$) with 10 random projections and blue for the random depth ($D$). Note that  all of them depict some notion of centrality despite selecting different curves. The random depth in this section is based on the distance of Lebesgue integrable functions with $n=m=500.$ Other values of $n$ and $m$ have been tried providing similar results. If the application at hand requires picking a regular curve as deepest element, it suffices to use, for instance, a Sobolev distance.

As commented above, in the deconvolution process, the elements corresponding to the curves whose integral over time is above the selected deepest curve are assigned their random depth value and  the rest of elements of the array are set either to 0 or ignored. To illustrate this,  Figure \ref{f2} shows the deconvolution based on the random depth from three different perspectives; sagittal, coronal and axial views (the co-ordinates number 50 of $Z,$ $Y$ and $X$ respectively).  The zero values are represented in blue and the positive values increase from yellow to red.

\begin{figure}[tbh]
\begin{center}
\hspace{-.7cm}
\includegraphics[width=.35\linewidth]{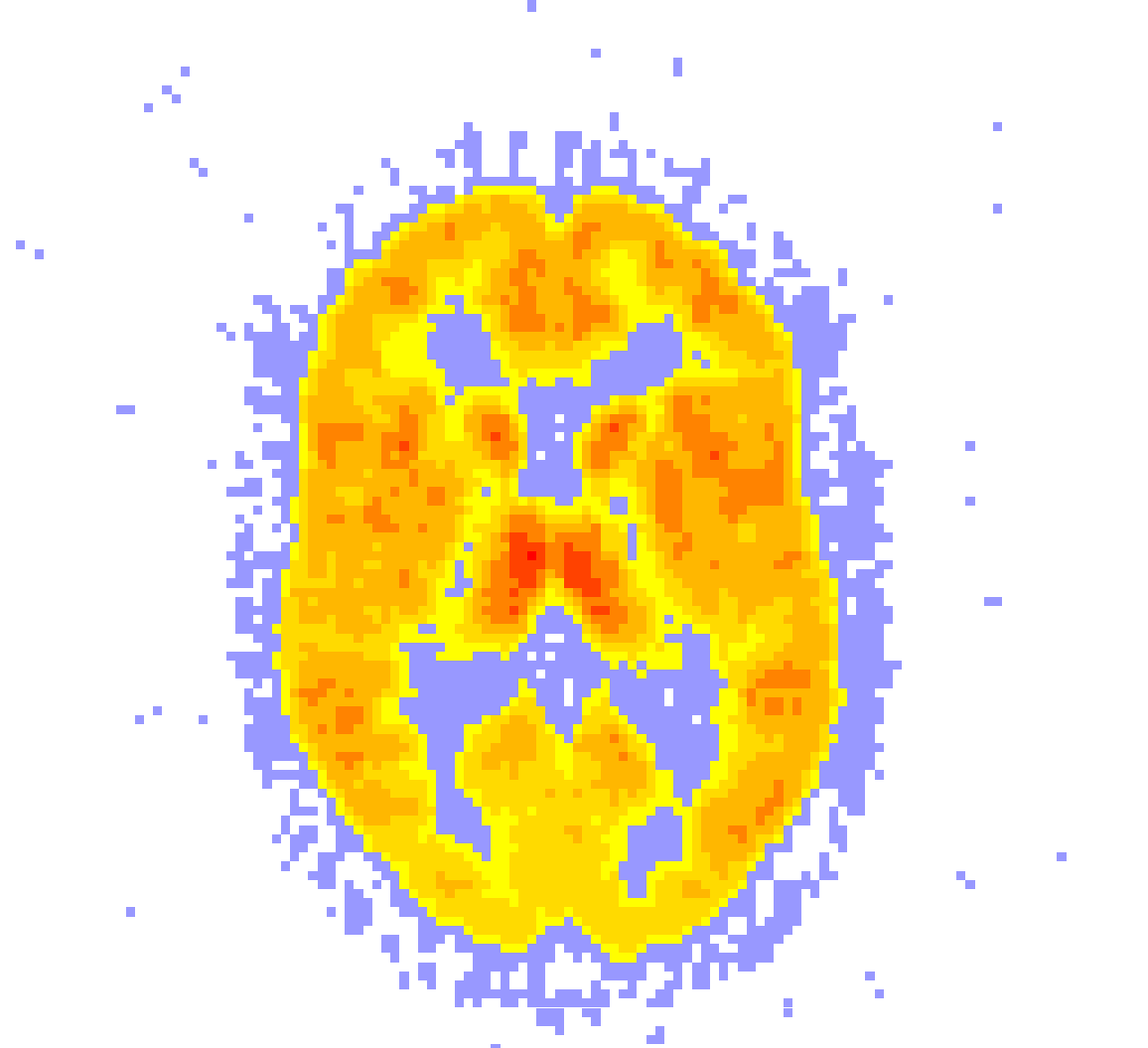}
\hspace{-1cm}
\includegraphics[width=.29\linewidth]{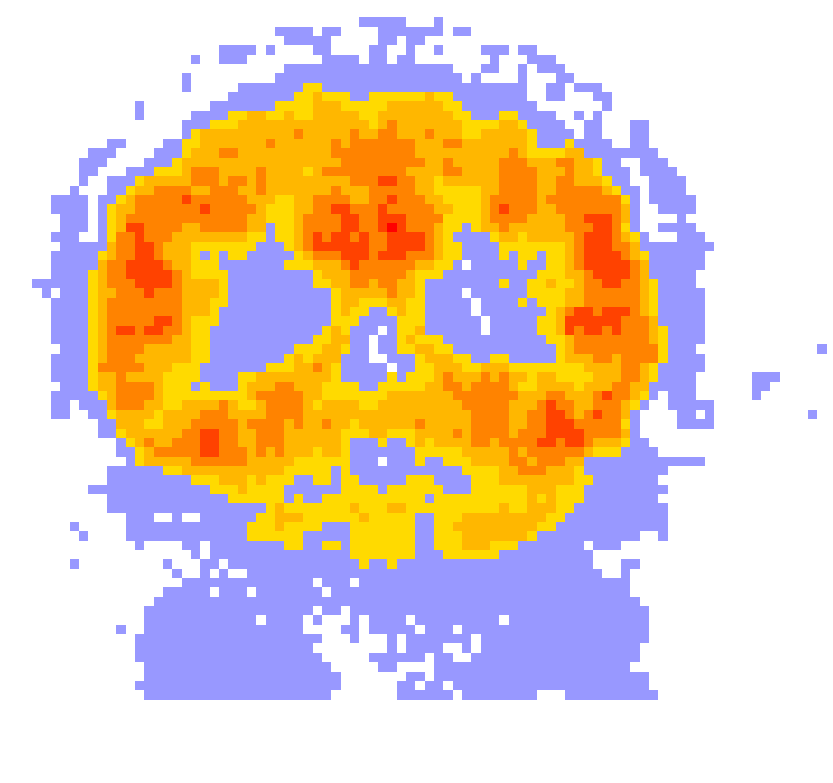}
\hspace{-.5cm}
\includegraphics[width=.348\linewidth]{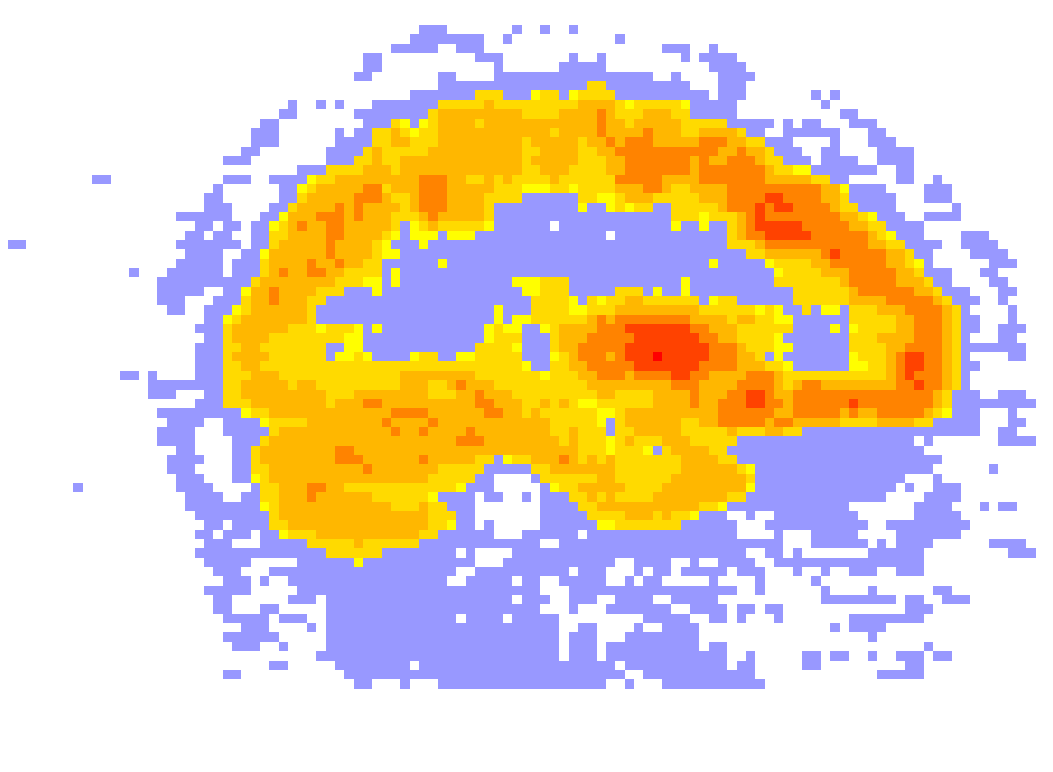}
\vspace{-.5cm}
\includegraphics[width=.1\linewidth,height=4cm]{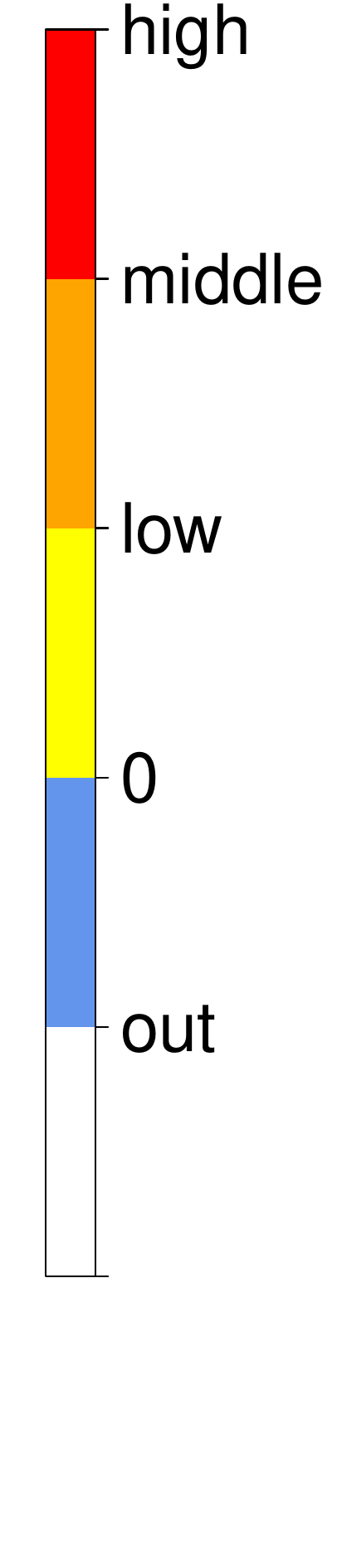}
\end{center}
\caption{$50$th section of the deconvolution using the random depth for three different perspectives; sagittal (left), coronal (middle) and axial (right) views. The zero values are represented in blue and the positive values increase from yellow to red.}
\label{f2}
\end{figure}

In order to study the performance of the proposed deconvolution procedure, we compute how the resulting image correlates with a purely parametric approach, in this case, the linear compartmental modelling approach \citep{CunnJ:93,Gunn01}, which is the default modelling approach taken in PET. In this case, we take the volume of distribution ($V_T$) as determined by PET Spectral Analysis from the assumption of a compartmental model system as the measure of interest.
Given a percentage $p\in(0,100),$ we first consider the voxels that have the $p\%$ highest positive $V_T$ values. Then, we compute the correlation between $V_T$ and our proposed methodology at these voxel locations.
In comparing the use of the different depth functions in our methodology, in Figure \ref{Cor} we plot, for any percentage $p\in(0,100),$ the absolute values of the correlation  when our methodology is based on $D_I$ (red), $D_M$ (red), $D_B$ (green), $D_{T}$ with 1 random projection (cyan) and with 10 random projections (orange) and $D$ (blue).
It is clearly seen in the figure that our proposed notion of depth, $D$ (blue) is more closely correlated with the compartmental model output. In particular, the region of most interest for most PET experiments is those areas with high $V_T$, and thus, it is worth mentioning that for $p\in(0,60)$ 
the correlation that makes use of the random depth takes values in $[.9438 , .9636]$ with mean $.9590$ and standard deviation $.0056$. This is not only very close to one but also very stable. 

This is reassuring in that a completely nonparametric approach is able to reproduce results from a parametric model, which overall is known to be a good fit in many locations. 
However, it has 
been shown in a number of papers that the assumptions used in many parametric modelling approaches for PET data can be difficult to justify \citep{FPCA, OSullivan:09, zanderigo2015model} at all locations, and thus having a purely nonparametric approach can be seen to be both useful when it gives similar results to the parametric models, but also useful to highlight areas of difference. 

It is also observed in the figure that $D_I$ (red) and $D_M$ (black)  have a poorer performance compared to $D$ and, as theoretically suggested in Section \ref{sectionFD}, $D_B$ (green) and $D_{T}$ (cyan) (orange) perform very poorly when applied to this type of data. The performance of $D_{T}$ with 1 random projection (cyan) is a bit better but is generally unstable because of being based on one random projection.

\begin{figure}[htb]
\includegraphics[height=6cm,width=.49\linewidth]{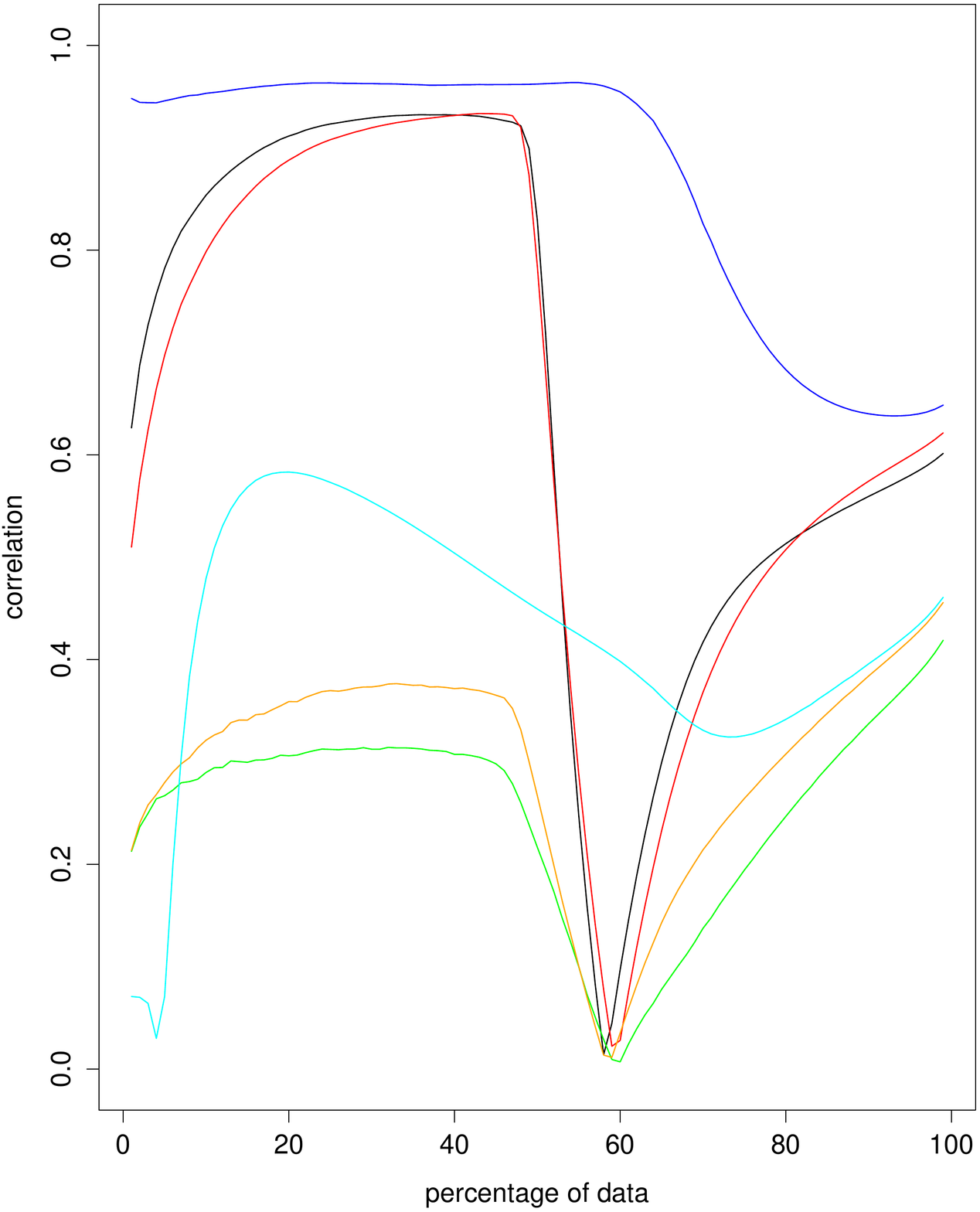}
\includegraphics[height=6cm,width=.49\linewidth]{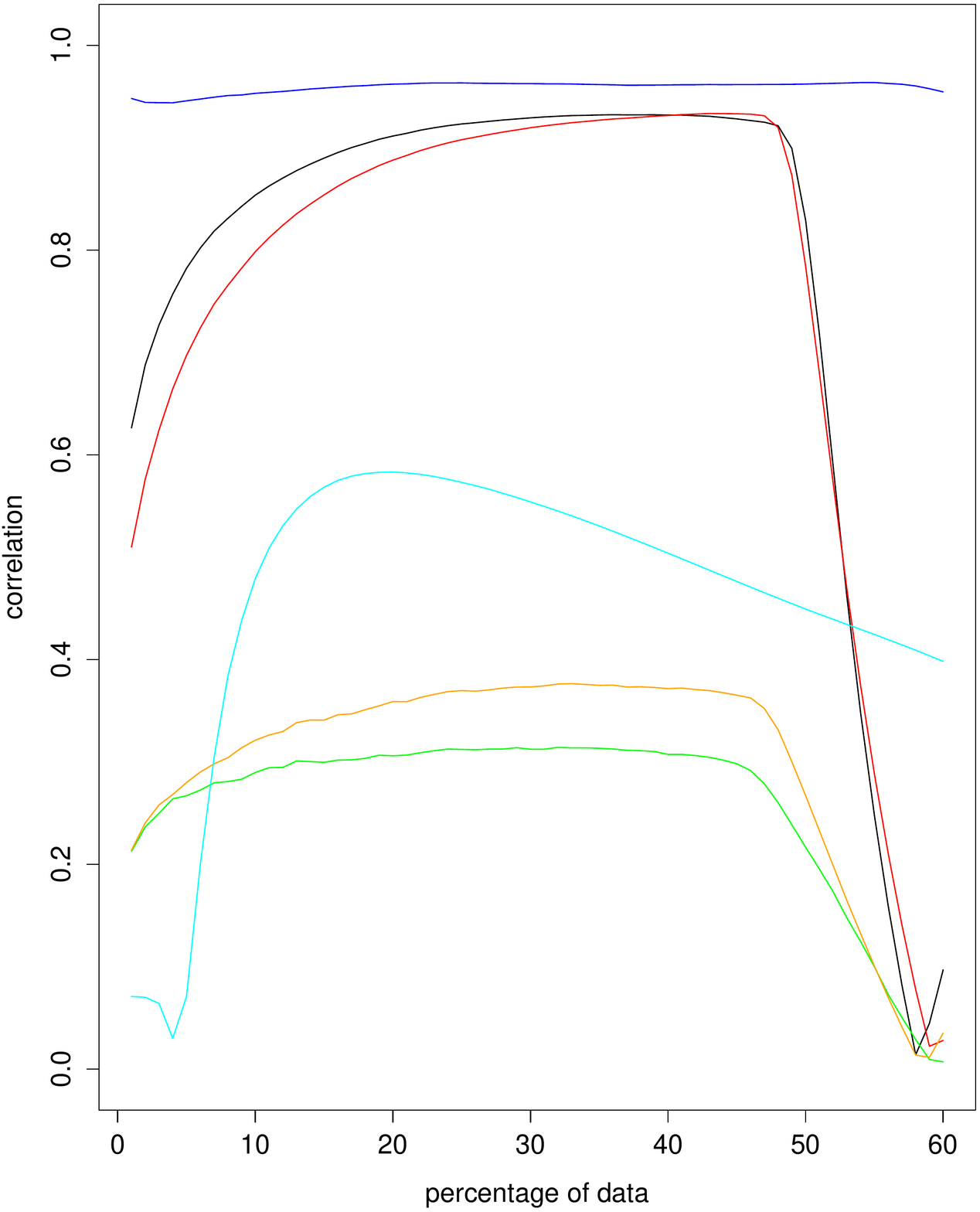}
\caption{Absolute values of the correlations between the linear model methodology and our proposed methodology, using $D_I$ (red), $D_M$ (black), $D_B$ (green), $D_{T}$ with 1 projection (cyan) and with 10 projections (orange) and $D$ (blue). In the left plot for $p\in(0,100)$ and in the right for $p\in(0,60).$ }
\label{Cor}
\end{figure}

\subsubsection{Test-retest reproducibility}

We aim to ensure that our methodology is reproducible. Thus, if a subject is scanned a second time, we need the result of applying our methodology to be highly correlated with the result obtained when applying it to the first scan. Furthermore, this correlation has to be higher than the average of correlating with scans of different subjects (intra subject variability needs to be lower than inter subject variability).
Moreover, we aim to obtain results that are more reproducible than the ones obtained with the existing non-parametric methodologies, such as
those based on functional principal component analysis (FPCA) \citep{FPCA}.

Here we analyse three different PET images, 1774, 1798 and 3715, which respectively are the retest scans of images 1680, 1794 and 3568, which we will respectively label $s_{i,k},$ $i=1,2,3$, $k=1,2$.
The analysis undertaken is explained through $s_{i,1}$ in the following steps. Let us denote by $r_{s_{i,k}}$ the result of applying our deconvolution procedure to $s_{i,k}$.
\begin{enumerate}
\item\label{1} Given a percentage $p\in(0,100),$ select the voxels of $r_{s_{i,1}}$  that constitute the $p\%$ of highest positive values in $r_{s_{i,1}}$.
\item\label{2} Compute the correlation between the values that $r_{s_{i,1}}$ and $r_{s_{i,2}}$ take at the voxels selected in \ref{1}.
\item\label{3} Do again \ref{1} and \ref{2} but interchanging the roles of $r_{s_{i,1}}$ and $r_{s_{i,2}}$. Then, compute the mean of the two correlations and denote it by $c_{s_{i,2}}$.
\item Do \ref{1}, \ref{2} and \ref{3} four more times, each time substituting $r_{s_{i,2}}$ by one of the following: $\{r_{s_{j,k}}: j\neq i, k=1,2\}.$ Define $m_{s_{i,1}}:=\mbox{mean}(\{c_{s_{j,k}}: j\neq i, k=1,2\}).$
\item Define $f_{s_{i,1}}(p):=(c_{s_{i,2}}-m_{s_{i,1}})/c_{s_{i,2}}.$
\end{enumerate}
In Figure \ref{Dif}, $f_{1774}$ (black), $f_{1798}$ (blue) and $f_{3715}$ (green) are plotted with respect to $p\in (0,100).$   The result of applying the same procedure when using the FPCA methodology of \cite{FPCA} is plotted using dotted lines. There are two observations from the figure, which indicates that the test-retest reproducibility is better achieved using the random  depth. The first is that the continuous lines (random depth) are always above the corresponding dotted line (the model based on FPCA).  The second, is that random depth always provides a higher correlation between scans of the same subject than of different subjects while this is not the case for the model based on FPCA with $f_{1774}$ (black). In this case,  it is clearly observed from the figure that on average the correlation between scans of the same subject is smaller than with different subjects, as the back dotted line takes negative values.

\begin{figure}[htb]
\begin{center}
\includegraphics[height=7cm,width=.4\linewidth]{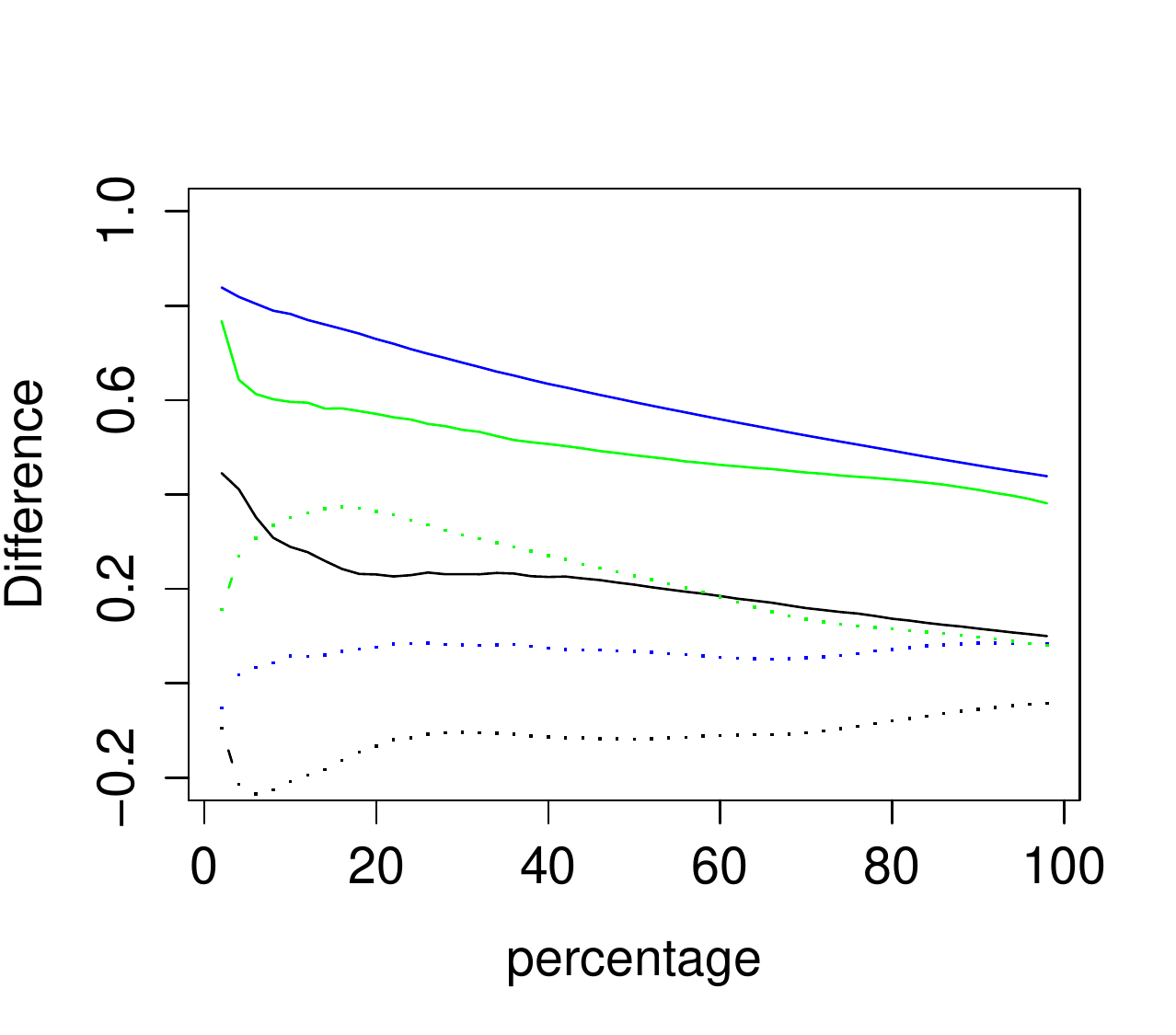}
\end{center}
\caption{Result of applying the test-retest procedure: $f_{1774}$ (black), $f_{1798}$ (blue) and $f_{3715}$ (green); for the proposed procedure  based on the random depth (continuous line) and the model based on FPCA (dotted line).}
\label{Dif}
\end{figure}


\subsection{Application to the selection of a representative subject of a dynamic neuroimaging dataset}\label{sectRep}
It is often useful to find a representative subject from a study. The dataset we use consists of nine 3D-images over time corresponding to six patients, three patients having being subject to two PET scans. The scans coming from one subject are in general  not independent. The estimator of the depth using our proposed notion is robust in this sense, thanks to the fulfillment of property P-6 that regards the  continuity in  probability (Theorem \ref{Td}). The methods used in this section are equally valid if the dataset consists of distinct patients.

In the left plot of Figure \ref{BA}, the raw data is displayed, particularly the sagittal view  for the $45$th slice of $Z$ and the $15$th time point. For each PET scan, we design an individual mask based on the thresholding procedure above. In the right plot of Figure \ref{BA} we  observe the image on the left after determining a mask. As can be seen, they all occupy different areas of the image, making direct naive comparison difficult.

\begin{figure}[htb]
\includegraphics[width=.4\linewidth]{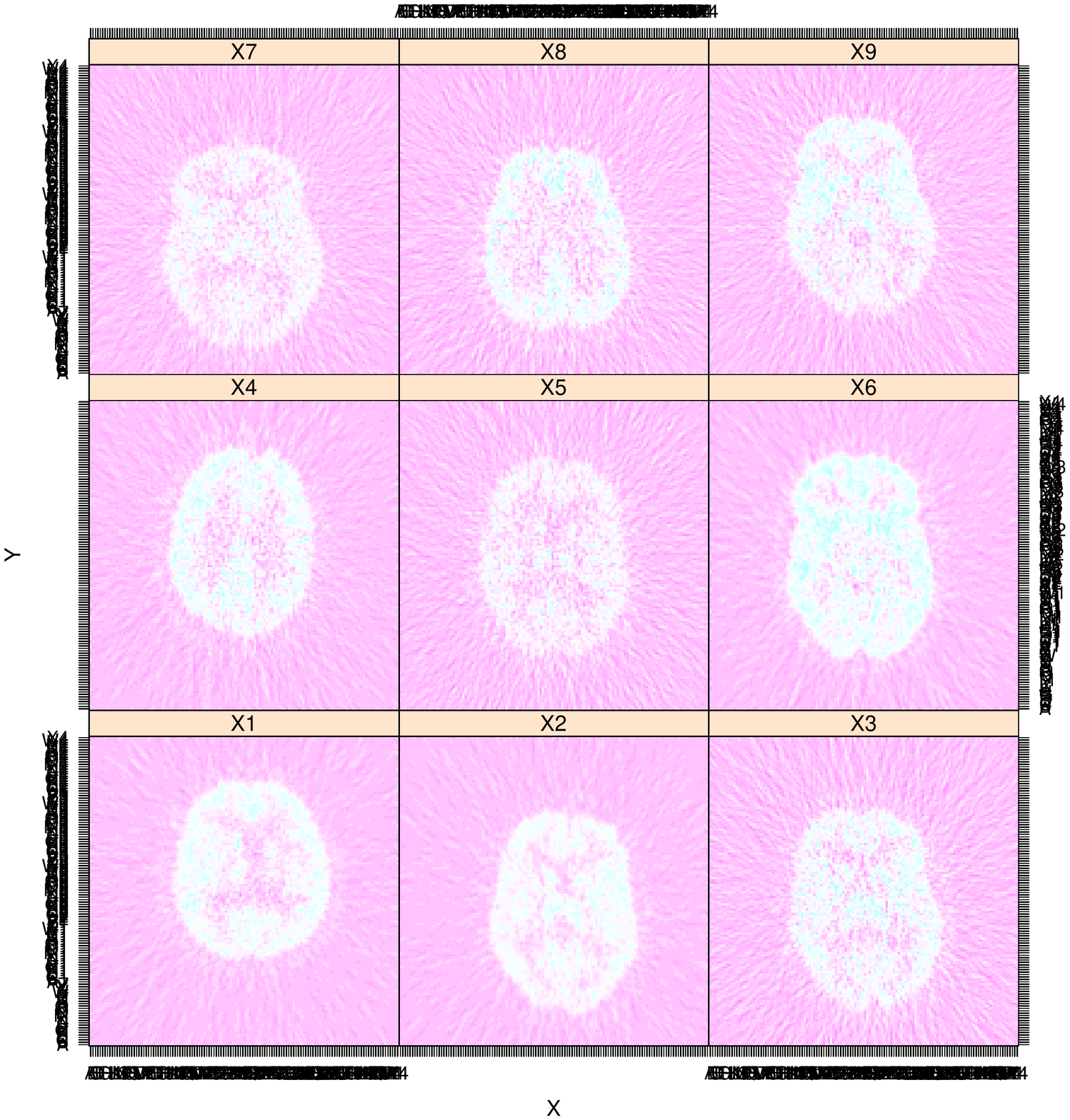} 
\includegraphics[width=.4\linewidth]{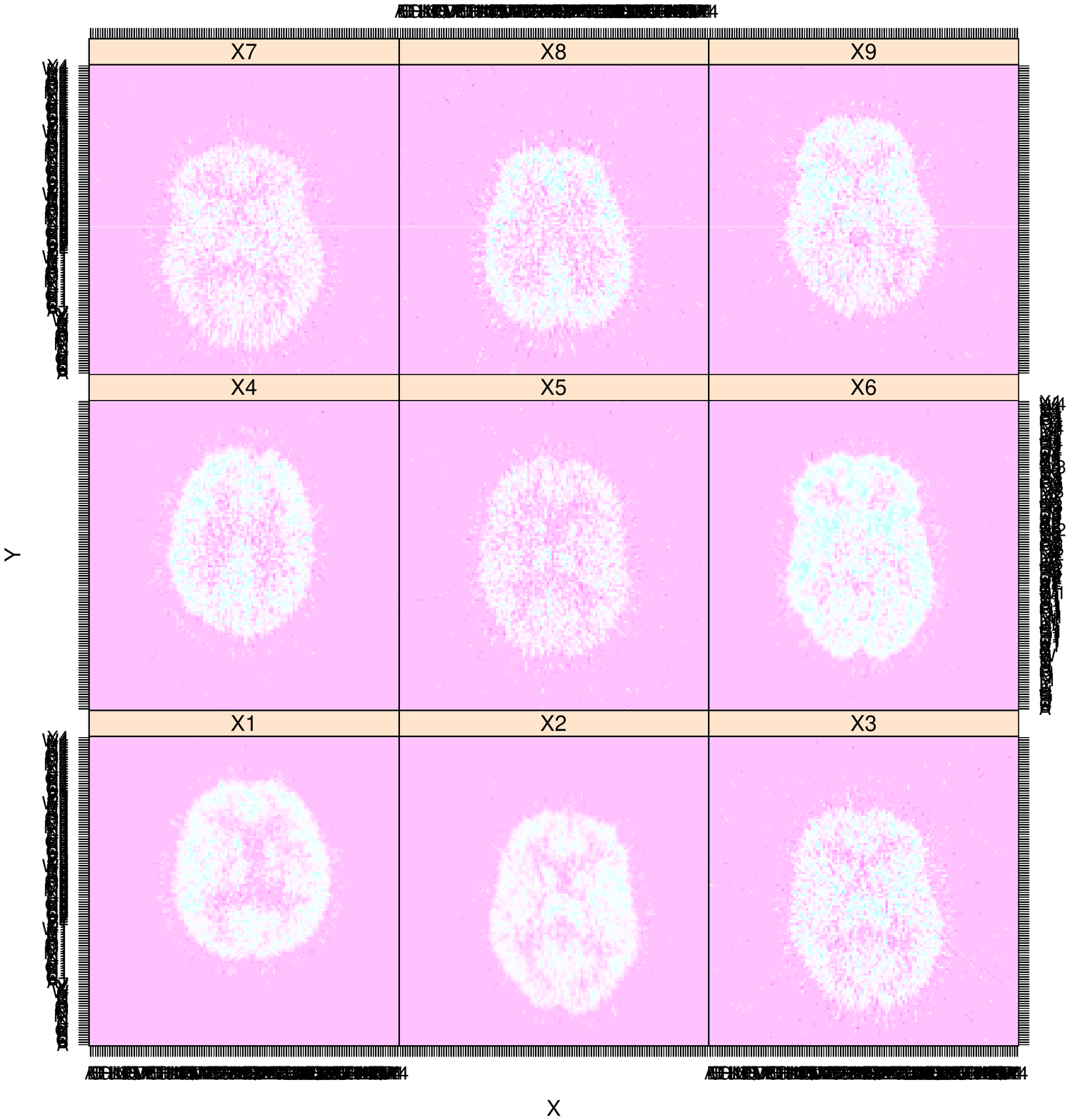} 
\caption{Sagittal view for the $45$th slice of $Z$ and the $15$th time point, for the nine elements of the dataset before applying a mask (left) and after applying it (right).}
\label{BA}
\end{figure}

For ease of presentation, after applying the mask, we denote each of the nine elements of the dataset  by $x_1, x_2, \ldots, x_9$. For each $j=1, \ldots, 9,$  $x_j(t)\in\RR$ for any $t\in I\subset \RR^4.$ 
As explained in Section \ref{sectionIn} in each PET image, a scan specific input function is required. To take into account this variation, each $x_j$ is normalised by the integral of the input function. We refer to this integral over time of the input as the intensity of the input.  In the left plot of Figure \ref{Input}, the input function of one of the scans ($x_1$) is shown. In  the right plot of the figure, we have plotted the image corresponding to those in Figure \ref{BA} but after also dividing each of the nine 3D images over time by its corresponding intensity of the input.

\begin{figure}[htb]
\includegraphics[height=7cm,width=.4\linewidth]{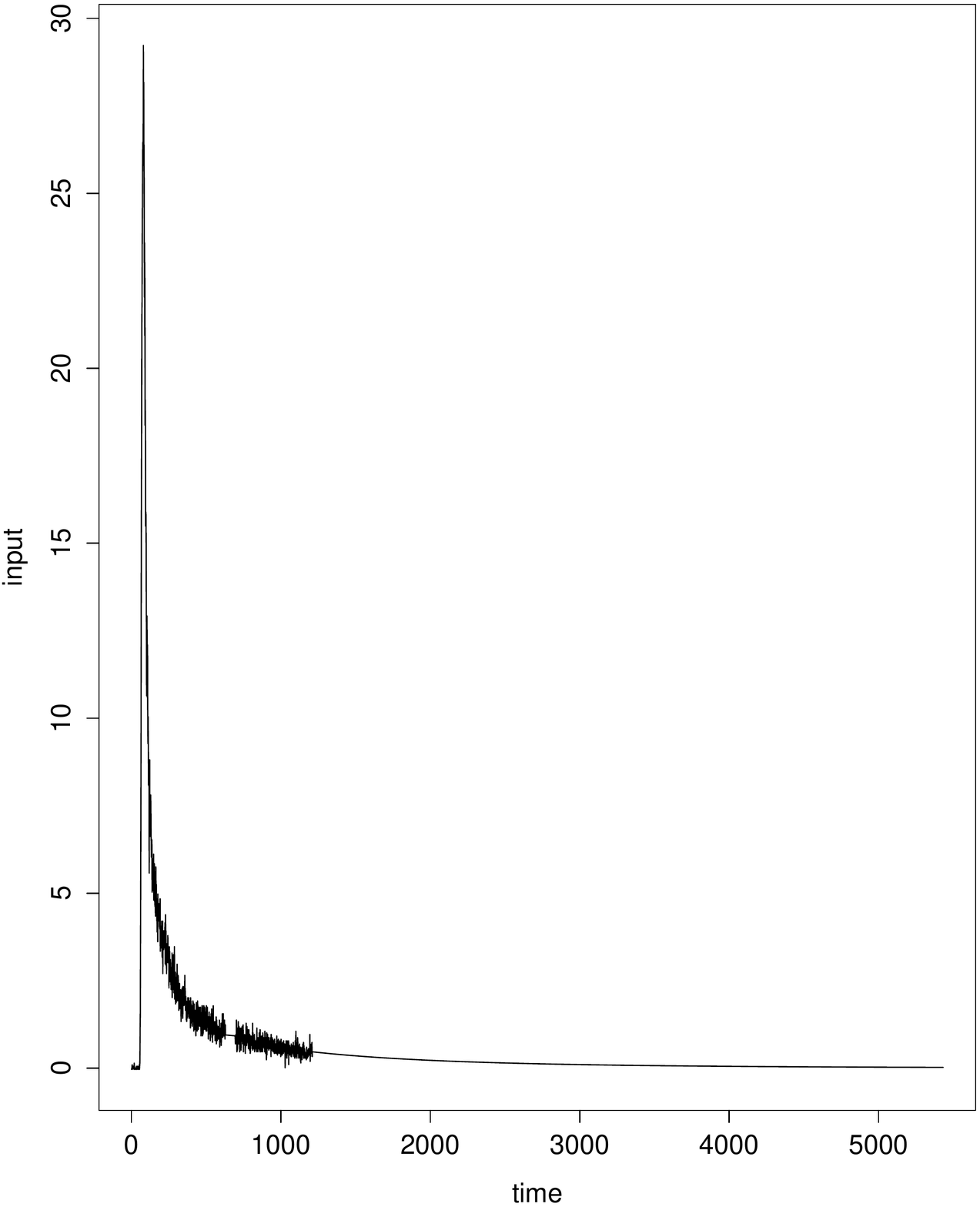}
\includegraphics[height=7.3cm,width=.4\linewidth]{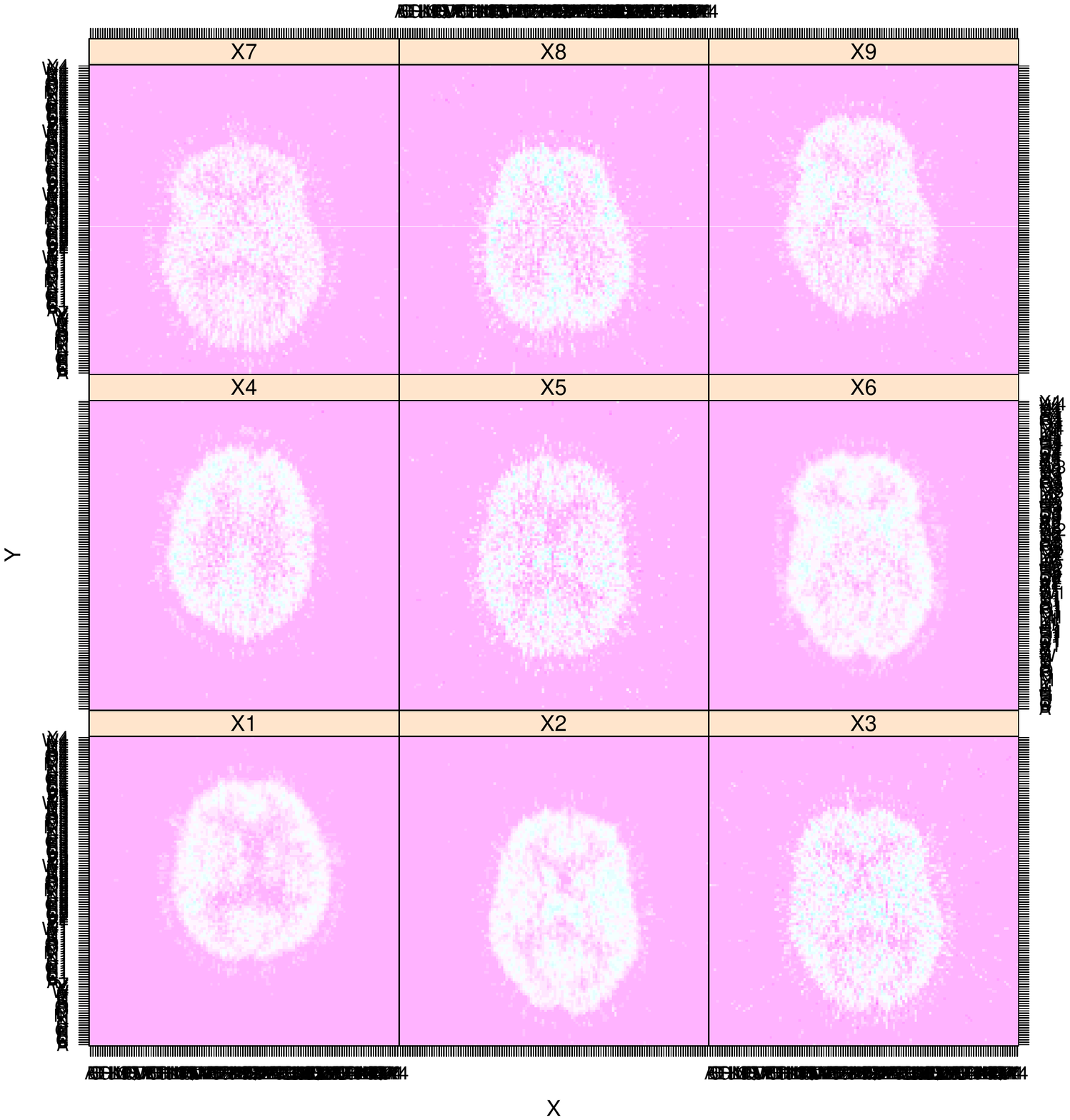}
\caption{Typical input function over time (left) and $XY$ perspective from Figure \ref{BA} after dividing the data by its corresponding integral of the input function (right).}
\label{Input}
\end{figure}

As commented in Section \ref{sectionFD}, our aim is to obtain a representative of the set by computing the metric depth function of $x_1,  \ldots, x_9$ with respect to the empirical distribution that corresponds to these nine functions,  using the pseudo-distance provided in Definition \ref{HD}. There is a difference between spatial dimensions, $X,$ $Y$ and $Z,$ and temporal one, $T.$ When deem appropriate, this difference is taken into consideration through the appropriate selection of metric. Here, however, the four dimensions have been considered equivalents.
We  obtain that $$(i_1, \ldots, i_9)=(.5224,.5934,.5015,.5659,.5131,.5153,.4109,.5662,.4651)$$ and so,
$i(x_7)< i(x_9)< i(x_3)< i(x_5)< i(x_6)< i(x_1)<i(x_4)< i(x_8)< i(x_2).$
In computing the median of this values we obtain that it corresponds to $x_6$ and, therefore, this is the deepest element, which can be thought of as a representative  PET image of the set. It turns out that this is the PET image studied in previous section, although the choice made there was due to it having previously been studied as the exemplar image in \cite{jiang2009smoothing, FPCA}, indicating that it was already somehow viewed as representative, even without a formal definition. It is also worth commenting that  $x_2=1680,$  $x_4=1794,$  and $x_8=3568$ are in the right side of  $x_6=2913,$ while their retest scans,  $x_3=1774,$  $x_5=1798,$  and  $x_9=3715,$ are on the left side, indicating that there may be a test-retest effect in this data.

%


\section{Conclusion}\label{sectionDisc}

We have presented both a theoretical and empirical investigation of notions of functional depth, particularly those which are computationally tractable in situations where there are a large number of curves, the so called \textit{big} functional data setting. We have introduced the notion of random depth, and shown that this has good performance relative to other functional depths both in simulations and in regard to an application of neuroimaging.

Random depth satisfies the six principles needed for a proper notion of functional depth and it is easily computed. This allows it to be used in settings where there are a large number of curves, in particular, as the choice of the size of subsets chosen can be tailored to the computational resources available.

In the application to PET data, it has been shown that a completely nonparametric idea of deconvolution based on depth can be both similar to a parametric interpretation yet be implemented without using any parametric assumptions. It can also be calculated on a voxelwise basis. The notion of functional depth can also be used to select a representative subject from an imaging data set in a principled fashion.

\vspace{12pt}

\bibliography{Biblio}



\newpage
\setcounter{page}{1}
 \renewcommand{\thesection}{\Alph{section}}
\setcounter{section}{0}

\end{document}